\documentclass{article} 
\usepackage[utf8]{inputenc} 
\usepackage[T1]{fontenc}    
\usepackage{booktabs}       
\usepackage{nicefrac}         
\usepackage{microtype}      
\usepackage[dvipsnames]{xcolor}         
\usepackage[round]{natbib}
\usepackage[colorlinks=true,citecolor=blue,urlcolor=blue,linkcolor=blue]{hyperref}
\usepackage[margin=1in]{geometry}
\usepackage{xspace}
\usepackage[shortlabels]{enumitem}
\usepackage{amssymb,amsmath,amsthm,bbm}
\usepackage{verbatim,float,url,dsfont}
\usepackage{algorithm,algorithmic}
\usepackage{mathtools,enumitem}
\usepackage{multirow,multicol}
\usepackage{ragged2e}
\usepackage{xr-hyper}
\usepackage{array}
\usepackage{stmaryrd,scalerel}

\usepackage{threeparttable}  
\usepackage{subcaption} 
\usepackage{scrextend} 
\usepackage{arydshln} 

\usepackage{tikz}
\usetikzlibrary{trees,positioning,shapes}

\newtheorem{assumption}{Assumption}

\def\E{\mathbb{E}}
\def\P{\mathbb{P}}

\def\ind#1{\mathds{1}\left\{#1\right\}}

\usepackage{chngcntr}
\usepackage{apptools}
\AtAppendix{\counterwithin{theorem}{section}}
\AtAppendix{\counterwithin{proposition}{section}}

\usepackage{mathtools}

\usepackage{mdframed}
\definecolor{light-gray}{HTML}{F7F7F7}
\definecolor{frame-color}{HTML}{CFCFCF}
\surroundwithmdframed[
backgroundcolor=light-gray,%
linecolor=frame-color%
]{minted}

\usepackage{mathtools}
\mathtoolsset{showonlyrefs}

\newtheorem{theorem}{Theorem}
\newtheorem{lemma}{Lemma}
\newtheorem{corollary}{Corollary}
\newtheorem{proposition}{Proposition}

\theoremstyle{definition}
\newtheorem{remark}{Remark}
\newtheorem{definition}{Definition}

\usepackage{chngcntr}
\usepackage{apptools}
\AtAppendix{\counterwithin{lemma}{section}}
\AtAppendix{\counterwithin{algorithm}{section}}

\newcommand{\thetalow}{\underline{\theta}}
\newcommand{\thetahigh}{\overline{\theta}}
\newcommand{\slow}{\underline{s}}
\newcommand{\shigh}{\overline{s}}

\def\E{\mathbb{E}}
\def\PP{\mathbb{P}}

\newcommand{\noPD}{\rm{noPD}}
\newcommand{\PD}{\rm{PD}}

\newcommand{\alphahat}{\hat{\alpha}}

\newcommand{\signalaware}{signal aware\xspace} 
\newcommand{\signalblind}{signal blind\xspace}
\newcommand{\stackutility}{$\mathbb{U}^*_1$\xspace}
\newcommand{\stackstrat}{commitment strategy\xspace}
\newcommand{\truthstrat}{\pi^s_{\mathrm{truthful}}}
\newcommand{\stratstrat}{\pi^s_{\mathrm{strategic}}}
\newcommand{\pdstrat}{\pi^p_{\PD}}
\newcommand{\nopdstrat}{\pi^p_{\noPD}}

\newcommand{\cber}{CBER}

\usepackage[]{color-edits}
\addauthor{td}{red}
\addauthor{mw}{magenta}
\addauthor{na}{blue}

\title{Privacy Can Arise Endogenously in an \\ Economic System with Learning Agents}  
\author{Nivasini Ananthakrishnan\thanks{equal contribution}, Tiffany Ding\footnotemark[1], Mariel Werner\footnotemark[1], \\
Sai Praneeth Karimireddy, Michael I. Jordan} 

\date{University of California, Berkeley} 

\begin{document}
\maketitle

\begin{abstract}
We study price-discrimination games between buyers and a seller where privacy arises endogenously---that is, utility maximization yields equilibrium strategies where privacy occurs naturally. In this game, buyers with a high valuation for a good have an incentive to keep their valuation private, lest the seller charge them a higher price. This yields an equilibrium where some buyers will send a signal that misrepresents their type with some probability; we refer to this as \emph{buyer-induced privacy}. When the seller is able to publicly commit to providing a certain privacy level, we find that their equilibrium response is to commit to ignore buyers' signals with some positive probability; we refer to this as \emph{seller-induced privacy}. We then turn our attention to a repeated interaction setting where the game parameters are unknown and the seller cannot credibly commit to a level of seller-induced privacy. In this setting, players must learn strategies based on information revealed in past rounds. We find that, even without commitment ability, seller-induced privacy arises as a result of reputation building. We characterize the resulting seller-induced privacy and seller's utility under no-regret and no-policy-regret learning algorithms and verify these results through simulations. 
\end{abstract}

\section{Introduction}

The question of how to define and preserve privacy in the age of big data and machine learning has been a topic of ongoing debate in the computer science and policy communities. The most widely accepted theoretical framework for privacy is the notion of differential privacy~\citep{dwork2014privacy}, which provides a rigorous mathematical definition of privacy as the ability to withstand membership inference attacks. That is, differential privacy ensures that the output of a computation obfuscates whether any individual's data was present in the input.

However, the practical implementation of differential privacy has been fraught with challenges. There has been significant debate around how to interpret the key privacy parameter $\varepsilon$ and how to choose it to achieve meaningful privacy protection~\citep{nissim2017bridging}. This is especially true when data is continuously collected from users (what does it mean to have a guarantee of $\varepsilon=1$ \emph{per day}?) This has also led to controversies where companies have claimed their algorithms are private, when in fact the chosen $\varepsilon$ value implies negligible actual privacy~\citep{tang2017privacy}. Further complicating matters, there are multiple variants and extensions of differential privacy---e.g. $(\varepsilon, \delta)$-DP~\citep{dwork2014privacy}, Reyni-DP~\citep{mironov2017renyi}, Gaussian-DP~\citep{dong2019gaussian}, etc. Each have their own parameters with differing interpretations and implications.

Moreover, a growing body of work argues that the public's understanding of privacy is quite different from this formal notion of differential privacy~\citep{solow2022information,ligett2023we}. While differential privacy focuses on membership inference, privacy is more commonly understood to mean preventing others from using one's data in ways that are misaligned with one's interests, such as price discrimination or other exploitative practices. 

This work seeks to provide a new perspective on privacy that bridges the gap between the theoretical computer science view and the public's intuitive understanding. We develop a game-theoretic model of privacy that allows us to analyze the effect of privacy choices on all the stakeholders. Additionally, the framework shows how to derive \emph{optimal} privacy mechanisms that balance the gain in privacy with loss of accuracy in order to maximize net profit.
In our model, a ``principal'' (e.g., a platform or seller) can observe signals from ``agents'' (e.g., users or buyers) and use this information to maximize its own profit, while the agents have an incentive to obfuscate their data to prevent exploitation. 

We focus on a price-discrimination setting involving interactions between buyers and sellers. 
We show that ``buyer-induced privacy'' behavior, which resembles randomized response, arises endogenously as an equilibrium strategy. Furthermore, we find that the seller is often better off committing to not observing the agents' data at all (``seller-induced privacy''), as the revenue loss from buyer-induced privacy can be substantial.

Finally, we extend our analysis to a dynamic setting where the principal is a learning agent who interacts with multiple agents over time. We demonstrate how a simple external auditing mechanism can implement the principal's commitment to privacy and lead to an equilibrium with endogenously arising privacy-preserving behavior.

Our results provide a new framework for understanding privacy that encompasses both the theoretical guarantees of differential privacy and the practical, user-centric notion of privacy. By modeling privacy as an emergent property of an economic system, we hope to offer insights that can inform the design of privacy-preserving platforms and policies, ultimately bridging the gap between theory and practice in this important domain.

\subparagraph{Motivating example.}
In the absence of regulation, online retailers may price discriminate based on information they have collected about past purchases of the customers. Some customers may be willing to pay more for a good than others, perhaps due to innate preferences for certain types of good or because they have more disposable income. The retailer wants to identify customers with higher valuations and charge them higher prices in order to maximize their revenue. 

Since customers are aware of the potential for price discrimination, they may engage in evasive action to protect their privacy. Customers may avoid choosing goods that signal their true preferences for less consequential purchases, e.g., a high-income customer choosing between an expensive water bottle that is slightly better than a cheaper option may opt to buy the cheaper bottle in an attempt to obscure their income status. This evasive action imposes a cost on the customer, who misses out on buying their truly preferred product, and also on the retailer, who would have preferred to sell the more expensive product.

What are the behaviors that arise at equilibrium? What if the seller can credibly commit to not price discriminate? How do these behaviors change in more realistic settings where game parameters are not known and strategies must be learned based on past interactions? These are questions we answer in this paper.

\subsection{Preview of contributions}

We introduce a price-discrimination game in Definition~\ref{pd-game} that involves buyers of two types---one with a high valuation and one with a low valuation of an item. A seller may potentially track buyers' signals that reveal their valuations. We characterize the perfect Bayes Nash equilibrium of this game in Theorem~\ref{thm:single stage equilibrium} and show that a buyer-induced privacy mechanism emerges in the equilibrium. That is, the buyer with a high valuation, with some probability, chooses an evasive action to appear to have a low valuation. 

We then introduce commitment ability for the seller wherein a seller can commit to not track buyers' signals with some probability. In the price-discrimination game with commitment, the equilibrium response (Corollary~\ref{cor:stackelberg}) results in seller-induced privacy, which obviates the need for buyer-induced privacy. That is, with some probability, the seller chooses to commit to respect privacy and voluntarily does not track signals. Due to this privacy commitment from the seller, it is optimal for buyers to truthfully report their type. We call this seller-induced privacy the ``\stackstrat'' and denote the resulting utility \stackutility. 

In Section \ref{sec:repeated_interaction}, we remove the seller's commitment ability but give buyers access to the seller's historical pricing. We model this as a repeated interaction between a seller and buyers with each buyer participating in only one round. 
The pricing history is used by buyers to construct the seller's ``reputation'' (i.e., an estimate of the probability of price discrimination), which buyers then use to inform their signaling strategy. We model the buyers as using a reputation construction procedure that satisfies a consistency condition given in Definition~\ref{def:consistentbeliefs}, which requires that the reputation is able to differentiate between sellers employing price-discriminating strategies and non-price-discriminating strategies. 
In Proposition~\ref{prop:existence_of_consistent_estimator}, we show the existence of such a reputation mechanism using the available history. 
We show that consistent reputation can yield seller-induced privacy (i.e., ignoring signals), depending on the model of the seller; we consider no-regret and no-policy-regret sellers. Our findings are:

\begin{enumerate}
    \item With a no-regret seller, there could be no seller-induced privacy. That is, the seller can use signals and price discriminate in every round and still be no-regret (Proposition~\ref{prop:regret min may not increase utility}).
    \item Regret minimization achieves strictly less average utility (asymptotically) than \stackutility (Proposition~\ref{prop:regretOddsStack}).
    \item Employing the \stackstrat in every round is a no-policy-regret algorithm for the seller (Proposition~\ref{prop:policyRegretStack}).
    \item Employing the \stackstrat in every round ensures the seller (asymptotically) an average utility of \stackutility. This the highest possible average utility achievable (asymptotically) in the repeated interaction (Proposition~\ref{prop:noBetterThanStack}).
\end{enumerate}

\subsection{Related work}
Our work sits at the intersection of many areas, ranging from classical economics to online learning.

There is a vast literature on \emph{privacy} in computer science studying mechanisms for notions of privacy such as differential privacy \citep{dwork2014privacy}. The mechanisms arising in our setting resemble mechanisms in these works. We observe local privacy (buyer-induced privacy) where users add noise to their data. We also observe central privacy (seller-induced privacy) where the platform ensures similar outcomes for different user data.

Literature in economics studies the economic implications of enacting privacy mechanisms (see~\cite{acquisti2016economics} for a survey). Within this body of work, there is a literature on privacy and \emph{price discrimination} (e.g., \cite{acquisti2004, conitzer2012, montes2015, fudenberg2006behavior}). We build on this work and extend to a setting that relaxes common-prior assumptions for buyers and sellers so that players must now devise strategies based on what they learn from repeated interactions.

In these repeated interactions, we observe the emergence of a \emph{reputation-based privacy mechanism}. This reputation, learned by buyers based on previous interactions, takes the place of the prior that is used in the single-interaction game. There are numerous papers in economics on reputation focusing on sellers' reputations for the quality of the proffered good~\citep{horner2002reputation, shapiro1983reputation, ely2003reputation}. We focus on seller's reputation for enacting price discrimination and analyze how this arises in an online learning framework. 

We also study the differences in behavior that arise from seller \emph{commitment}, which has been studied in \cite{hart1988contract}, \cite{acquisti2004}, \cite{fudenberg2006behavior} and \cite{ichihashi2020online}. We show that even without commitment, similar behavior can arise through repeated interactions where reputation substitutes for the role of commitment. 

Finally, we draw upon work on \emph{online learning} and \emph{repeated games}. There are a number of papers~\citep{camara2020prior, deng2019strategizing, haghtalab2024calibrated, foster1997calibrated} on repeated interactions between a principal and an agent where the agent chooses actions based on evolving beliefs about the principal's actions. In our setting, we interpret the evolving beliefs as the reputation of the principal. Our setting differs in two ways. The first is that the principal's actions are not revealed at the end of the round. Instead partial information about the action, depending on the agent's response, is revealed. The second is that our results hold for weaker conditions on the agent's beliefs compared to previous work. 

\section{A Price-Discrimination Game} \label{sec:single-interaction} 

We formulate price discrimination as a sequential, incomplete-information game between $n$ buyers and a seller.

\begin{definition}[PD game]\label{pd-game} The price-discrimination game with parameters $n, \alpha, \mu, \thetahigh, \thetalow, c_B, c_S$, denoted the \emph{($n, \alpha, \mu, \thetahigh, \thetalow, c_B, c_S$)-PD game}, has the following extensive-form representation. 
\begin{enumerate}
    \item \textbf{Nature's move.} The game begins with Nature assigning types to each participant according to random draws. For $i \in [n]$, the type for buyer $i$ is $\theta_i \in \{\thetalow, \thetahigh\}$, representing their valuation of the item being sold, with $\thetalow < \thetahigh$. A buyer is type $\thetahigh$ with probability $\mu$ and type $\thetalow$ with probability $1-\mu$. The seller's type $\chi$ is either \emph{\signalaware} ($\chi = 1$) or \emph{\signalblind} ($\chi = 0$). The seller is \signalaware with probability $\alpha$ and \signalblind with probability $1-\alpha$. 
    \item \textbf{Signaling stage.} Based on their assigned type $\theta_i$, each buyer signals $s_i \in \{\slow, \shigh\}$. Signaling one's true type ($\slow$ for type $\thetalow$ and $\shigh$ for type $\thetahigh$) incurs no cost, whereas signaling a mismatched type, referred to as ``evasion,'' imposes a cost $c_B$ on the buyer and a cost $c_S$ on the seller.\footnote{We can more generally allow for each type of buyer impose a different evasion cost (e.g., if a  $\thetahigh$-buyer evades, the costs are $\bar{c}_B, \bar{c}_S \in \mathbb{R}$, and if a $\thetalow$-buyer evades, the costs are $\underline{c}_B, \underline{c}_S  \in \mathbb{R}$. However, as we later show, the only costs that are relevant are the evasion costs associated with the $\thetahigh$-seller, because the $\thetalow$ seller will never choose to evade, so we can think of $c_B = \bar{c}_B$ and $c_S = \bar{c}_S$.}
    \item \textbf{Pricing decision.} The seller chooses a price $p_i$ to set for buyer $i$. The information the seller can use to set the prices depends on the type of seller. A signal-aware seller can set prices depending on the signals sent by the buyers, that is, they can set one price for all buyers that signaled $\slow$ and a different price for all buyers that signaled $\shigh$. A signal-blind seller must set the same price for all buyers since they have no information to distinguish buyers.
    \item \textbf{Purchase decisions.} Each buyer, based on the price $p_i$ set for them and their valuation $\theta_i$, makes a choice $b_i \in \{0, 1\}$, to purchase the item ($b_i=1$) or not ($b_i=0$).
    \item \textbf{Utilities.} All players receive their respective utilities. Each buyer's positive utility is zero if they do not buy the item and the difference between their valuation and price otherwise. If they took evasive action in the signaling stage, their negative utility is equal to their cost of evasion $c_B$. That is, buyer $i$'s utility is 
        \begin{align} \label{eq:buyer_utility_function}
            u_B(\theta_i, s_i, p_i, b_i) &= (\theta_i - p_i)b_i - c_B e(\theta_i, s_i) 
        \end{align}
    where $e(\theta_i, s_i) = \mathds{1}\{(\theta_i = \thetalow \wedge s_i = \shigh) \vee (\theta_i = \thetahigh \wedge s_i = \slow)\}$ indicates evasion or not.
    The seller's overall utility is the sum of utilities $u_S(\theta_i, s_i, p_i, b_i)$ from their interactions with each buyer. The positive utility due to buyer $i$ is the revenue $p_i$ if buyer $i$ buys and zero otherwise. If the buyer took evasive action in the signaling stage, the seller incurs negative utility $c_S$.
        That is, the seller's utility is
    \begin{align}
            u_S \left ((\theta_i, s_i, p_i, b_i)_{i=1}^n \right ) &= \sum_{i=1}^n u_S(\theta_i, s_i, p_i, b_i) 
            = \sum_{i=1}^n p_i b_i - c_S e(\theta_i, s_i) \label{eq:seller_utility_function}.
        \end{align}
\end{enumerate}
\end{definition}

\subparagraph{Mixed strategies.} For simplicity of presentation, our game definition is stated in terms of pure strategies (i.e., players take deterministic actions). However, we can more generally allow players to employ mixed strategies. A \emph{mixed strategy} for a player is a distribution over allowed actions conditioned on the information available when taking the action: buyer $i$'s mixed signaling strategy induces a conditional distribution over signals $\pi_i^s(\cdot | \theta_i) \in \Delta(\{\slow, \shigh\})$; the seller's mixed pricing strategy induces conditional distributions $\pi^p(\cdot | \slow, \chi)$, $\pi^p(\cdot | \shigh, \chi)$ over positive reals with the constraint $\pi^p(\cdot | s = \slow, \chi = 0) = \pi^p(\cdot | s = \shigh, \chi = 0)$; finally, each buyer $i$'s mixed buying strategy induces conditional distribution $\pi_i^b(\cdot | \theta_i, p_i) \in \Delta(\{0, 1\})$.

Let $\pi = (\pi^s, \pi^p, \pi^b)$ denote a mixed strategy profile. $\pi$, along with the probability of player types described in Step 1 of Definition \ref{pd-game} (which we will denote 
$p(\chi)$ and $p(\theta_i)$)
induce a distribution over action profiles with the probability of an action profile $\left ( \chi, (\theta_i, s_i, p_i, b_i)_{i=1}^n \right )$ given by
\begin{align} \label{eq:joint_distribution}
    \mathbb{P}\left ( \chi, (\theta_i, s_i, p_i, b_i)_{i=1}^n \right ) &= p(\chi) \prod_{i=1}^n p(\theta_i) \pi_i^s(s_i | \theta_i) \pi^p(p_i | \theta_i, \chi) \pi_i^b(b_i | \theta_i, p_i).
\end{align}
Given a mixed strategy profile $\pi$,
we will denote the expected utility for the seller and buyer $i$ by
\begin{align*}
    U_S(\pi) &= \E \left [ u_S \left ((\theta_i, s_i, p_i, b_i)_{i=1}^n \right ) \right ] \qquad \text{and} \qquad U_B^i(\pi) = \E \left [ u_B \left (\theta_i, s_i, p_i, b_i \right ) \right ]
    ,
\end{align*}
where the expectation is over the joint distribution in \eqref{eq:joint_distribution}. 

\textbf{Solution concept}. We study the \emph{perfect Bayes Nash equilibrium (PBNE)}. Mixed strategies of players constitute a PBNE if the following conditions hold: (1) sequential rationality, meaning that each player's strategy constitutes a best response to their beliefs about the other players' types and strategies, given the history of the game up to the point of choosing the action and (2) consistency of beliefs, meaning that players' beliefs about other players' types are updated following Bayes' rule.

The following theorem characterizes the PBNE of the price-discrimination game described in Definition \ref{pd-game}. 

\begin{theorem}\label{thm:single stage equilibrium}
An ($n, \alpha, \mu, \thetahigh, \thetalow, c_B, c_S$)-PD game has the following unique perfect Bayes Nash equilibrium. Define $\Delta\theta=\thetahigh-\thetalow$.
\begin{enumerate}[(a)]
    \item Buyers with type $\theta_i = \thetalow$ will signal $s_i=\slow$.
    \item Buyers with type $\theta_i = \thetahigh$ will signal 
    \begin{align}
        s_i =
        \begin{cases}
            \slow \text{ w.p. }  q^* & \text{if } \alpha  > c_B/\Delta\theta \quad \qquad \text{where } q^*=\min\bigg\{1,\frac{(1-\mu)\thetalow}{\mu\Delta\theta}\bigg\}
            \\
            \shigh & \text{otherwise}.
        \end{cases}
    \end{align}
    \item The signal-aware seller sets price 
    \begin{align}
    p^*_{\mathrm{\signalaware}}(s)=
    \begin{cases}
        \thetalow & \text{if signal } s=\slow \text{ is observed}\\
        \thetahigh & \text{if signal } s=\shigh \text{ is observed}.
    \end{cases}
    \end{align}
    \item The signal-blind seller sets price 
    \begin{align}
    p^*_{\mathrm{\signalblind}}=
    \begin{cases}
        \thetalow & \text{if } \thetalow \geq \mu \thetahigh \\
        \thetahigh & \text{otherwise}.
    \end{cases}
    \end{align}
    \item Buyer $i$ buys the good if and only if their price $p_i$ is at most their value, so 
    \[b_i = \mathbbm{1}\{\theta_i \leq p_i\}.\]
\end{enumerate}
\end{theorem}
The proof is given in Appendix \ref{sec:proof of single stage equilibrium}. 

\begin{remark}[Buyer-induced privacy]
    The $\thetahigh$-buyers' equilibrium response can be interpreted as a privacy-protecting mechanism. This type of buyer is vulnerable to price discrimination, so rather than always signaling their true type, they may choose to randomize their signal. More specifically, if the cost of evasion is very high, the $\thetahigh$-buyer will tell the truth, but if the evasion cost is low enough, the $\thetahigh$-buyer can receive a reduction in price that is higher than their evasion cost. In the latter case, the $\thetahigh$-buyer must then choose the maximum evasion probability $q^*$ such that it is still in the seller's best interest to take the the buyer's signal at face value. We call this randomization ``buyer-induced privacy.''
\end{remark}

Theorem \ref{thm:single stage equilibrium} tells us that strategic behavior can only happen if $c_B < \Delta\theta$ (otherwise, we can never have $\alpha > \nicefrac{c_B}{\Delta \theta}$, so buyers will always signal truthfully). For the rest of the paper, we will focus on this setting.

\begin{assumption}
In all following results, we assume $c_B < \Delta\theta$.
\end{assumption}

A natural next question is how each player's utility is affected by the game parameters. In particular, we focus on the effect of $\alpha$, due to its connection to privacy. In Figure \ref{fig:order of utilities}, we visualize the utilities of the seller and $\thetahigh$-buyers as $\alpha$ varies from 0 to 1. Observe that the seller's utility increases for $\alpha$ less than some threshold value $\alpha^*$, whose exact value we give in the corollary below. This corresponds to the set of PD-games where the buyer's equilibrium response is truthful. Beyond $\alpha^*$, the $\thetahigh$-buyers' equilibrium response changes to being strategic and the seller's utility drops. 
We formalize the ordering of utilities in the following corollary. 

\begin{corollary}\label{cor:order of utilities}(Order of utilities) Fix $n, \mu, \thetahigh, \thetalow, c_B, c_S$ and let $u_S(\alpha), u_B(\alpha)$ denote the seller's and $\thetahigh$-buyers' equilibrium utilities of the $(n, \alpha, \mu, \thetahigh, \thetalow, c_B, c_S)$-PD game. $u_S(\cdot)$ is maximized at $\alpha^* = c_B / \Delta \theta$, and the equilibrium utilities for the settings where the seller is always \signalblind ($\alpha=0$), is always \signalaware ($\alpha=1$), and is signal aware with probability $\alpha^*$ ($\alpha=\alpha^*$) have the following orderings:
\begin{enumerate}[(a)]
    \item When $\thetalow \geq \mu\thetahigh$,
    \begin{align}
        u_S(\alpha^*)
        >
        u_S(0)
        >
        u_S(1)
        \qquad \text{and} \qquad u_B(0)
        >
        u_B(1)
        =
        u_B(\alpha^*).
    \end{align}
    \item When $\thetalow<\mu\thetahigh$,
    \begin{align}
        u_S(\alpha^*)
        >
        u_S(0)
        >
        u_S(1)
        \qquad \text{and} \qquad 
        u_B(1)
        >
        u_B(0)=
        u_B(\alpha^*).
    \end{align}
\end{enumerate}
$\thetalow$-buyers always receive a utility of zero, regardless of the value of $\alpha$.
\end{corollary}

\begin{figure}[!t]
\centering
  \centering
  \includegraphics[width=5in]{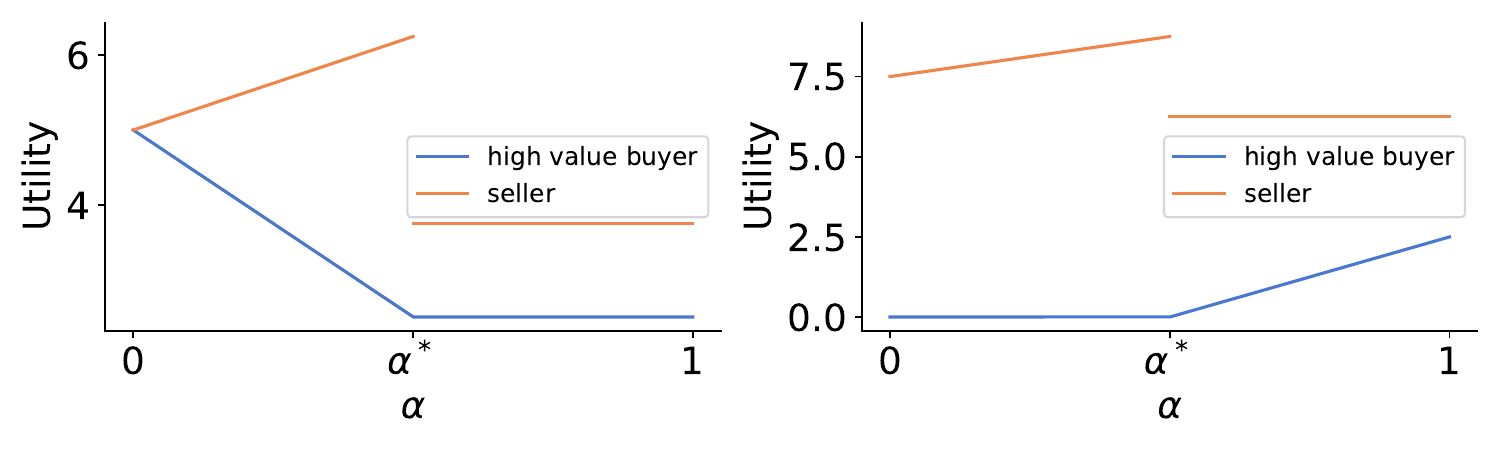}
  \caption{Plots of the $\thetahigh$-buyer and seller utilities as a function of $\alpha$ in the $\thetalow \geq \mu \thetahigh$ setting (left) and the $\thetalow < \mu \thetahigh$ setting (right).}
  \label{fig:order of utilities}
\end{figure}

The proof is given in Appendix \ref{sec: proof of order of utilities}.

\subsection{Price discrimination with seller commitment} 
A key takeaway from Corollary \ref{cor:order of utilities} is that the seller's utilities are dependent on the value of $\alpha$, and if the seller could choose a value of $\alpha$, they would want to choose $\alpha=\alpha^*$ to maximize their utility. Suppose we are now in a setting where the seller is able to choose and publicly commit to an $\alpha$. As a motivating example, suppose that the seller must go through a data broker to access signals, and the data broker publishes trusted summaries of what fraction of buyers the seller requests data on. In such a setting, where $\alpha$ is chosen by the seller instead of treated as given, we arrive at the following equilibrium.

\begin{corollary}(Equilibrium of price-discrimination game with commitment) \label{cor:stackelberg} When the seller has commitment power (i.e., is able to credibly communicate to sellers that they will not price discriminate with some probability), the perfect Bayes Nash equilibrium of the PD-game consists of the following strategies:
\begin{enumerate}[(a)]
    \item The seller commits to not price-discriminating (by playing $p^*_{\mathrm{\signalblind}}$ from Theorem \ref{thm:single stage equilibrium}) with probability $1-\alpha^*$, where $\alpha^* = c_B / \Delta \theta$. 
    \item All buyers always signal truthfully.
\end{enumerate}
The buyers' buying decisions are the same as in Theorem \ref{thm:single stage equilibrium}.
\end{corollary}
\begin{proof}
    (a) follows directly from Corollary \ref{cor:order of utilities}, which tells us that the seller's utility is maximized at $\alpha^*$, and (b) comes from applying Theorem \ref{thm:single stage equilibrium} with $\alpha=\alpha^*$.
\end{proof}

\begin{remark}
    Commitment ability allows the seller to achieve a higher utility by providing seller-induced privacy. This seller-induced privacy obviates the need for buyers to take evasion action to create buyer-induced privacy, which benefits the seller.
    We use \stackutility to refer to the seller's maximum achievable equilibrium utility in the single interaction price discrimination game with commitment. This utility is achieved when the seller plays the strategy given in Corollary \ref{cor:stackelberg}.
\end{remark}

\section{Repeated Interactions} \label{sec:repeated_interaction}

In the previous section, we saw the emergence of seller-induced privacy when the seller has commitment ability. If possible, the seller would commit to providing seller-induced privacy (by ignoring signals with probability $1 - \alpha^*$, as in Corollary~\ref{cor:stackelberg}), thereby limiting the extent of price discrimination performed by the seller. However, these results hinge on the buyer believing that the $\alpha$ stated by the seller truly corresponds to the probability of price discrimination. Without this credible commitment from the seller, the story becomes more complicated. 

In this section, we study whether seller-induced privacy can still arise in the absence of such commitment ability, through the development of a reputation based on the seller's historical pricing. 
We ask the question of how the extent of privacy and resulting utilities differ under reputation-based privacy versus commitment-based privacy.
We model the seller as making pricing decisions using an online learning algorithm and show how different models such as \emph{no-regret} and \emph{no-policy-regret} lead to different answers to this question.  

In the repeated interaction setting, we also relax the assumptions that the distribution $\mu$ over agent types and the probability $\alpha$ that the seller looks at the agent's signal are publicly known. 
Rather than playing the single-interaction equilibrium strategies, which require full knowledge of game parameters, 
the players now have to learn strategies online based on past interactions. 

\subsection{Setup}

We consider repeated interactions between a seller and buyers where a new batch of buyers is drawn at each round. We call this as the \emph{repeated PD protocol}. Each round is similar to the one-shot PD-game from Definition~\ref{pd-game} but with the following differences: (1) There is one fixed seller throughout all rounds. 
(2) When players choose actions, they not only have access to information from the current round (as was the case in the one-shot PD game) but also some information from previous rounds. 
Specifically, at round $t$, 
the seller has access to $((s_i^{\tau}, p_i^{\tau})_{i=1}^n)_{\tau=1}^{t-1}$, the signals they observed and the prices they set in previous rounds, and each buyer $i$ has access to $(((\theta_i^{\tau}, s_i^{\tau}, p_i^{\tau})_{i=1}^n)_{\tau=1}^{t-1})$, the buyer types, signals, and prices of all buyers from previous rounds. 
This modeling of the buyers' access is appropriate in settings where buyer information is pooled either through crowd-sourcing or by an auditing entity and made available to buyers. 
(3) The parameter $\mu$ (the probability of a type-$\thetahigh$ buyer) is not known to the seller.
(4) The probability that the seller will price discriminate is not known to buyers, as was assumed in the one-shot PD game; rather, buyers must estimate this probability based on past rounds. We write out the repeated interaction protocol in detail in Algorithm~\ref{alg:repeated_pd}.

\subsection{Model of the buyers}
Since each buyer participates in only one round of the repeated PD protocol, the equilibrium response is still appropriate to model the buyer's response.
However, in the repeated interaction setting, we no longer assume the buyers hold a static, prior belief about the probability of a signal-aware seller. Instead, buyers have evolving beliefs based on the seller's interactions with past buyers.

Some specific buyer strategies we will refer to are $\truthstrat$, which corresponds to always signaling truthfully, and $\stratstrat$, which corresponds to signaling $\slow$ with probability $q^*$ (as defined in Theorem \ref{thm:single stage equilibrium}) and signaling $\shigh$ with probability $1-q^*$.
We consider the following model of buyer behavior.

\begin{definition}[Consistent belief based equilibrium responding (\cber) buyers] Consistent belief based equilibrium responding buyers (or \cber-buyers) form a sequence of beliefs $(\alphahat_t)_{t=1}^T$ satisfying a consistency property defined below and at round $t$, choose the corresponding equilibrium strategy (from Theorem \ref{thm:single stage equilibrium}) of the PD-game with $\alpha = \alphahat_t$. That is, $\thetalow$-buyers always signal truthfully, and $\thetahigh$-buyers signal truthfully (play $\truthstrat$) if $\alphahat_t \leq \alpha^*$ and signal the opposite type with probability $q^*$ otherwise (play $\stratstrat$).   
\end{definition}

We now explain the consistency property.
Given a sequence of seller mixed strategies action profiles that induce the sequences of distributions $(\pi^p_t(\cdot | s=\shigh))_{t=1}^T$ and $(\pi^p_t(\cdot | s=\slow))_{t=1}^T$ indicating price distributions at each round for signals $\slow, \shigh$ respectively, define $\alpha_t$ to be
\[\alpha_t = \mathbb{P}_{\overline{P} \sim \pi^p_t(\cdot | s=\shigh), \underline{P} \sim \pi^p_t(\cdot | s=\slow)} \left [ \overline{P} \neq \underline{P} \right ]. \]
That is, $\alpha_t$ denotes the probability of a different price for $\shigh$ compared to $\slow$ at round $t$. The probability here is over the randomness due to the seller's  mixed strategy at round $t$. $\alpha_t$ is a measure of extent of price discrimination by the seller at round $t$.

\begin{definition}[Consistent sequence]\label{def:consistentbeliefs} Let $\bar{\alpha}_T = (1/T)\sum_{t=1}^T \alpha_t$.
We say a sequence of estimators $(\alphahat_t)_{t=1}^T$ is \emph{consistent} if 
$\lim_{T \rightarrow \infty} \left \lvert \E [\alphahat_T] -  \bar{\alpha}_T \right \rvert = 0$,
where the expectation is taken over the randomness of the history $H_T = ((\theta_i^{t}, s_i^{t}, p_i^{t})_{i=1}^n)_{t=1}^{T-1}$
used to construct $\alphahat_T$. 
\end{definition}

A useful implication of consistency is that $\alphahat_T$ converges pointwise to $\bar{\alpha}_T$.

\begin{lemma} \label{lemma:consistency_implication}
    If $(\alphahat_t)_{t=1}^T$ is a consistent sequence of beliefs, then for any $\epsilon<0$ and $\delta>0$, there exists some positive integer $N$ such that for all $T > N$, we have
    $
    \P\left[\left|\alphahat_T - \bar{\alpha}_T\right| \geq \epsilon\right] \leq \delta
    $.
\end{lemma}
\begin{proof}  
Due to consistency and the definition of limits, there exists $N$ such that for all $T > N$, we have $|\E[\alphahat_T] -  \bar{\alpha}_T| \leq \delta\epsilon$. Thus, for $T > N$, we can apply Markov’s inequality to get $\P(|\alphahat_T - \bar{\alpha}_T| \geq \epsilon) \leq (|\E[\alphahat_T] -  \bar{\alpha}_T|)/\epsilon \leq \delta\epsilon/\epsilon =\delta$.
\end{proof}

The following proposition and associated proof provide an algorithm to construct a consistent sequence of estimators $(\alphahat_t)_{t=1}^T$.

\begin{proposition}[Existence of consistent sequence] \label{prop:existence_of_consistent_estimator}
Assume that buyers equilibrium-respond to $\alphahat_t$ at each round $t$. Then, for any sequence of seller actions, there exists a sequence of estimators $(\alphahat_t)_{t=1}^T$ that is consistent.  
\end{proposition}

\emph{Proof sketch:} 
Since there are multiple buyers at each round, we can infer whether the seller is price discriminating or not by comparing the prices charged to a buyer who signals $\slow$ and a buyer who signals $\shigh$. However, only some rounds are informative about price discrimination; in rounds where all buyers send the same signal, we are not able to determine if the seller had a price discriminatory pricing policy in place.
The consistent estimator $\alphahat_t$ we consider is the fraction of past rounds where price discrimination is observed, normalized to account for the probability that a round is likely to be informative about price discrimination. We show that $E[\alphahat_t]= (1/t) \sum_{\tau =1}^{t-1}\alpha_{\tau}$, which implies that 
$\lim_{T \rightarrow \infty} \left \lvert \E [\alphahat_T] -  \bar{\alpha}_T \right \rvert 
    = \lim_{T \rightarrow \infty} \left \lvert (1/T) \sum_{t=1}^{T-1} \alpha_t -  (1/T) \sum_{t=1}^T \alpha_t \right \rvert 
    = \lim_{T \rightarrow \infty} \alpha_T/T
    = 0 $. 
See Appendix~\ref{sec:proof of existence of consistent estimator} for the full proof.
 
\subsection{Model of the seller}
Since the seller does not a priori know the distribution over buyer types and is engaged in multiple rounds of the repeated interaction, modeling the seller's response by the one-shot equilibrium from Theorem \ref{thm:single stage equilibrium} is not reasonable. 
Instead, we consider the seller as optimizing various common objectives of repeated interactions such as regret minimization and policy-regret minimization. 

The seller's mixed strategy at a given round is a pair of probability distributions $\pi^p_t = (\pi^p_t(\cdot|\shigh), \pi^p_t(\cdot|\slow))$. Let $\Pi$ denote the set of possible mixed strategies. For rational sellers, we can focus on distributions supported only on $\{\thetalow, \thetahigh\}$ without loss of generality. 
Prices supported on $\{\thetalow, \thetahigh\}$ maximize seller revenue in each round. The seller's effect on future rounds is also not affected by limiting the support. This is because the parameters $\alpha_t$ that the buyers' consistent estimator estimates treats \emph{any} difference in prices as indicating price discrimination, so all price differences are treated the same.  

Some specific seller strategies we will refer to are $\pi^p_{\PD}$ and $\pi^p_{\noPD}$. The former is the ``always-price-discriminating strategy,'' with $\pi^p_{\PD}(\thetahigh | \shigh) = \pi^p_{\PD}(\thetalow | \slow) = 1$. The latter is the ``never-price-discriminating strategy,'' with 
$\pi^p_{\noPD}(\thetalow | \slow) = \pi^p_{\noPD}(\thetalow | \shigh) = 1$ if $\thetalow \geq \mu \thetahigh$ and $\pi^p_{\noPD}(\thetahigh | \slow) = \pi^p_{\noPD}(\thetahigh | \shigh) = 1$ otherwise.

\subsubsection{Regret-minimizing seller}

The first seller model we consider is a regret-minimizing seller. 

\begin{definition}[Seller's regret] 
    Given a sequence of mixed strategy profiles $\{\pi_t\} = \{(\pi^s_t, \pi^p_t, \pi^b_t)\}_{t=1}^T$, the \emph{seller's average regret} is 
      \[R_T^S(\{\pi_t\}_{t=1}^T) = \frac 1 T \left [ \max_{\pi^{p*} \in \Pi} \sum_{t=1}^T U_S(\pi^s_t, \pi^{p*}_t, \pi^b_t)  - \sum_{t=1}^T U_S(\pi^s_t, \pi^p_t, \pi^b_t) \right ].\]
\end{definition}

\begin{definition}[No-regret algorithm]
Let $\mathcal{A}_B$ be an algorithm employed by the buyer in the repeated PD protocol. A seller algorithm $\mathcal{A}_S$ in the repeated PD protocol is a \emph{no-regret algorithm} for the seller given $\mathcal{A}_B$ if the sequence of mixed strategies $(\pi_t)_{t=1}^T$ generated by the interaction between $\mathcal{A}_B$ and $\mathcal{A}_S$ has seller's average regret that is sublinear in the number of rounds. That is, $R^S_T((\pi_t)_{t=1}^T) \in o(1).$
\end{definition}

We will denote by $(\pi_t)_{t=1}^T$ the sequence of random variables denoting the players' mixed strategies in each round. Our results analyze the asymptotic convergence of average seller utility. We say that the average seller utility \emph{asymptotically converges} to some value $v$ if   
$\lim_{T \rightarrow \infty} \E \left [(1/T) \sum_{t = 1}^T U_S(\pi_t) \right ] = v$. We write 
$U_S(\pi^p)$ and $U_S(\pi^s, \pi^p)$
when it is clear what the other arguments are.

If the seller employs a no-regret algorithm, then the seller could end up always price-discriminating i.e., no seller-induced privacy. This is stated below.
\begin{proposition}\label{prop:regret min may not increase utility}(Always price-discriminating is regret minimizing) 
Given \cber-buyers, the seller algorithm that always employs the price-discrimination strategy i.e., $\pi^p_t = \pdstrat$ for all timesteps $t$ is a no-regret algorithm for the seller. The seller's average utility asymptotically converges to a value at most $u_S(1)$,
where $u_S(1)$ is the seller's equilibrium utility in the single-interaction PD-game with $\alpha=1$. 
\end{proposition}
\begin{proof}[Proof sketch]

The strategy of \cber-buyers in each round is either $\truthstrat$ or $\stratstrat$.
For both these buyer responses, the seller's optimal strategy is to always price discriminate, as shown in the computation of the seller's equilibrium response in the proof of Theorem \ref{thm:single stage equilibrium}. 
In other words, the seller incurs zero regret in each round by always price-discriminating.

Next, we analyze the seller's average utility. Note that when $\pi^p_t = \pdstrat$, the probability of seeing different prices for different signals is $\alpha_t=1$, so $\bar{\alpha}_t = 1$ for all $t$. By Lemma \ref{lemma:consistency_implication}, $\alphahat_t$ becomes greater than $\alpha^*$ eventually (where $\alpha^*$ is as defined in Corollary \ref{cor:stackelberg}), which causes $\thetahigh$-buyers to play $\stratstrat$. 
In other words, eventually the seller and buyers will all be playing their equilibrium strategies for the PD-game with $\alpha=1$, so their average utilities will converge to the corresponding equilibrium utilities. 
See Appendix~\ref{proof:regret min may not increase utility} for the full proof.
\end{proof}

The next proposition tells us that regret minimization necessarily causes the seller to achieve a worse expected average utility that the optimal utility they can achieve in the single interaction setting.

\begin{proposition}[Regret minimization is inherently at odds with achieving \stackutility]\label{prop:regretOddsStack}
    Given \cber-buyers, 
    for any no-regret seller algorithm,
    the seller's average utility asymptotically converges to strictly less than \stackutility. 
\end{proposition}
\begin{proof}[Proof sketch]
Define $\mathcal{T} = \{t \in [T]: \alphahat_t \leq \alpha^*\}$ to be the set of rounds
where $\thetahigh$-buyers' signaling strategy is $\truthstrat$. In all other rounds, their signaling strategy is $\stratstrat$. Define $\beta = (1/T) \sum_{t \in \mathcal{T}} \alpha_t$ to be a measure of simultaneous truthfulness from buyers and price-discrimination by the seller. Our proof involves the following parts. We outline the parts and state them as lemmas here and prove them in Appendix~\ref{proof:regretOddsStack} 

\begin{enumerate}
    \item Obtaining \stackutility requires the buyers to be truthful strictly more than $\alpha^*$ fraction of rounds.
    \begin{lemma}\label{lem:lowTruthfulLowUtility}
        $\lim_{T \rightarrow \infty} |\mathcal{T}| / T \leq \alpha^*$ implies that $\lim_{T \rightarrow \infty} \left (\sum_{t=1}^T U_S(\pi_t) \right ) / T < \text{\stackutility}$.
     \end{lemma}

     \item The no regret property requires that the seller price discriminates in most rounds where buyers are truthful. So $\beta$ is close to $|\mathcal{T}| / T$.

      \begin{lemma}\label{lem:ImplicationNoRegret}
    $\lim_{T \rightarrow \infty} {|\mathcal{T}|} / T \leq \lim_{T \rightarrow \infty}   {\sum_{t \in \mathcal{T}} \alpha_t} / T$. 
    \end{lemma}

    \item There is a limit on simultaneous price-discrimination and truthful signaling due to the buyers' consistent beliefs. That is, $\beta$ converges to at most $\alpha^*$.
    \begin{lemma}\label{lem:consistencySimultaneousPDTruthful}
        $\lim_{T \rightarrow \infty} \sum_{t \in \mathcal{T}} \alpha_t / T \leq \alpha^*.$
    \end{lemma}
\end{enumerate}

From Lemmas~\ref{lem:ImplicationNoRegret},~\ref{lem:consistencySimultaneousPDTruthful}, $\lim_{T \rightarrow \infty} |\mathcal{T}| / T \leq \alpha^*$. Lemma~\ref{lem:lowTruthfulLowUtility} shows that this means average seller utility is strictly less than \stackutility.
  
\end{proof}

\subsubsection{Policy-regret-minimizing seller}
As we have seen, regret minimization does not guarantee that the seller achieves higher than price-discrimination utility. On the other hand, if we model the seller as minimizing policy regret \citep{arora2020policy}, the seller \emph{necessarily} achieves utility that is higher than the utility achieved by the naive strategy of always price discriminating. 

\begin{definition}[Seller's policy regret]
     Consider a buyer algorithm $\mathcal{A}_B$ and a seller algorithm $\mathcal{A}_S$. Let $(\pi_t(\mathcal{A}_B, \mathcal{A}_S))_{t=1}^T$ be the sequence of mixed strategies generated by the interaction between $\mathcal{A}_B$ and $\mathcal{A}_S$. 
     Given a sequence of mixed strategies $(\pi_t)_{t=1}^T$, the \emph{seller's average policy regret} of $(\pi_t)_{t=1}^T$ relative to a buyer algorithm $\mathcal{A}_B$ and a baseline class $\mathbb{A}_S$ of seller algorithms is 
    \[PR^S_T \left ((\pi_t)_{t=1}^T; \mathcal{A}_B, \mathbb{A}_S \right) = \max_{\mathcal{A}_S \in \mathbb{A}_S} \frac 1 T \sum_{t=1}^T U_S( \pi_t(\mathcal{A}_B, \mathcal{A}_S)) - \frac 1 T \sum_{t=1} U_S(\pi_t) \]
\end{definition}

\begin{definition}[No-policy-regret algorithm]
Let $\mathcal{A}_B$ be an algorithm employed by the buyer in the repeated PD protocol. An algorithm $\mathcal{A}_S$ is a \emph{no-policy-regret algorithm} for the seller given $\mathcal{A}_B$ and relative to a class of seller algorithms $\mathbb{A}_S$ if the sequence of mixed strategies $(\pi_t(\mathcal{A}_B, \mathcal{A}_S))_{t=1}^T$ generated by the interaction between $\mathcal{A}_B$ and $\mathcal{A}_S$ satisfies $PR^S_T((\pi_t(\mathcal{A}_B, \mathcal{A}_S); \mathcal{A}_B, \mathbb{A}_S)_{t=1}^T) \in o(1).$
\end{definition}

Consider a baseline class $\mathbb{A}_S^{MS}$ consisting of seller algorithms that employ the same mixed strategy in each round, that is, $\pi_t^p(\cdot|\shigh)$ is the same distribution for all $t$ and similarly for $\pi_t^p(\cdot|\slow)$.
\begin{proposition}[Policy-regret-minimizing seller achieves \stackutility]\label{prop:policyRegretStack}
Given \cber-buyers, 
if the seller achieves sub-linear policy regret relative to $\mathbb{A}^{MS}_S$, then 
the seller's average utility asymptotically converges to at least \stackutility. 
\end{proposition}
\begin{proof}[Proof sketch]
    Under the conditions of this proposition, the seller's utility must, by definition of policy regret, approach a utility at least as high (or better) than the utility of any strategy in $\mathbb{A}_S^{MS}$ as $T \to \infty$. Recall that \stackutility is the seller utility achieved in the PD game when $\alpha = \alpha^*$.
    Consider the PD game that results in a seller utility of at least $\text{\stackutility}- \epsilon$, which is achieved by the seller price-discriminating with probability $\tilde{\alpha} < \alpha^*$. Then the repeated-interaction strategy of always price-discriminating with probability $\tilde{\alpha}$ has an average expected utility of at least $\text{\stackutility}- \epsilon$ (this must be true due to the consistency of buyer beliefs; see the full proof in Appendix \ref{sec:proof_of_policyRegretStack} for details). Taking $\epsilon$ to 0 gives the desired result.
\end{proof}

Combining the previous result with the following result tells us that a no-policy regret seller's algorithm will cause the seller's average utility to asymptotically converge to \emph{exactly} \stackutility. In fact, this result tells us the stronger result that there does not exist \emph{any} seller algorithm that can achieve utility higher than \stackutility.

\begin{proposition}\label{prop:noBetterThanStack}
    Given \cber-buyers, for any seller algorithm, the seller's average utility asymptotically converges to at most \stackutility. 
\end{proposition}

\begin{proof}[Proof sketch]
    This proof is similar to the argument of the proof of Proposition~\ref{prop:regretOddsStack} and the full proof is in Appendix~\ref{proof:noBetterThanStack}. The key ideas again are that for high seller utility, there must be sufficiently many rounds where simultaneously, the seller price discriminates and the buyer reports truthfully. Since the buyers' belief estimators are consistent, this cannot be the case. The difference between the average seller utility and \stackutility is a constant times the following quantity:
    $\frac 1 T \sum_{t \in \mathcal{T}} (\pi^p_t(\thetahigh | \shigh) - \pi^p_t(\thetalow | \slow)) - \alpha^*$, where $\mathcal{T}$ is the set of rounds where the buyer signals truthfully.
    Lemma~\ref{lem:relnPDMeasures},~\ref{lem:consistencySimultaneousPDTruthful} (from the proof of Proposition~\ref{prop:regretOddsStack}) show that the consistency property implies that this difference converges to most zero. 
\end{proof}

\section{Experiments}
In this section, we simulate Algorithm \ref{alg:repeated_pd} with $\mu=0.5$, $\thetalow=5$, $\thetahigh=15$, $c_B=c_S=5$, and $n=10$
and empirically verify our theoretical claims from Section \ref{sec:repeated_interaction}. We report the convergence of buyer and seller utilities, seller actions, and buyer estimators. The seller and buyer algorithms we consider are described below. Code is available at
\url{https://github.com/nivasini/PrivacyDynamics}.

\subsection{Algorithms}
\subparagraph{Seller}
\begin{enumerate}
    \item \textbf{Signal-blind seller.} The seller plays the regret-minimizing Exp3 algorithm (specifically Exp3-IX in Chapter 12 of \cite{lattimorebandits}). At round $t$ the seller sets a price $p_t \in \{\thetalow,\thetahigh\}$ according to the algorithm's current sampling distribution, charges $p_t$ to all buyers and updates the sampling distribution based on the resulting average utility from the buyers' purchase decisions.
    \item \textbf{Signal-aware seller.} The seller plays a contextual version of Exp3, which we call CExp3, in which the algorithm maintains two sampling distributions over prices $\{\thetalow,\thetahigh\}$, conditioned on the received signal, $\slow$ or $\shigh$. At each round, the seller samples once from each distribution and charges one price $\underline{p}_t$ to all buyers who signal $\slow$ and $\overline{p}_t$ to all buyers who signal $\shigh$. Depending on the sampling distributions, $\underline{p}_t$ and $\overline{p}_t$ may or may not be equal.
    \item \textbf{Stackelberg equilibrium seller.} The seller commits to an $\alpha^*=\nicefrac{c_B}{\Delta\theta}$ level of price-discrimination, i.e., they play the $(\alpha=1)$-PD equilibrium strategy (Theorem \ref{thm:single stage equilibrium}) with probability $\alpha^*$ and the $(\alpha=0)$-PD equilibrium with probability $1-\alpha^*$.
\end{enumerate}
\subparagraph{\cber-Buyer.} Using a sequence of consistent estimators $\{\alphahat_{\tau}\}_{\tau=1}^{t-1}$ (Def. \ref{def:consistentbeliefs}) to estimate the seller's probability of price-discrimination at each round, each buyer plays the $(\alpha=\alphahat_{\tau})$-PD equilibrium strategy. For our simulations, buyers use the estimator described in \eqref{eq:consistent-estimator} to estimate the seller's probability of price discrimination at each round. All buyers in a single round use the same estimator.

\subsection{Discussion}
\subparagraph{Convergence of Utilities.}
Figure \ref{fig:convergence of utilities} shows convergence of seller and buyer utilities for each of the seller's algorithms played against a \cber-buyer. As expected, when a seller plays Exp3 (which ignores signals) against a \cber-buyer, the players' utilities converge to the $(\alpha=0)$-PD equilibrium utility (Theorem \ref{thm:single stage equilibrium}). When the seller plays CExp3 (which observes signals) against a \cber-buyer, the seller's utility converges to the $(\alpha=1)$-PD equilibrium utility. Given our experiment parameters, multiple different distributions $\pi^p_t(\cdot | s=\slow)$ reward the seller equivalently, while some are more favorable for the buyer than others. Therefore, while the seller's utility will always converge to $(\alpha=1)$-PD, the buyer's utility may converge to something less than $(\alpha=1)$-PD. Finally, when the seller plays the Stackelberg equilibrium against a \cber-buyer, the players' utilities converge to the $(\alpha=\alpha^*)$-PD equilibrium utility.

\begin{figure}[!t]
  \centering
  \includegraphics[scale=0.4]{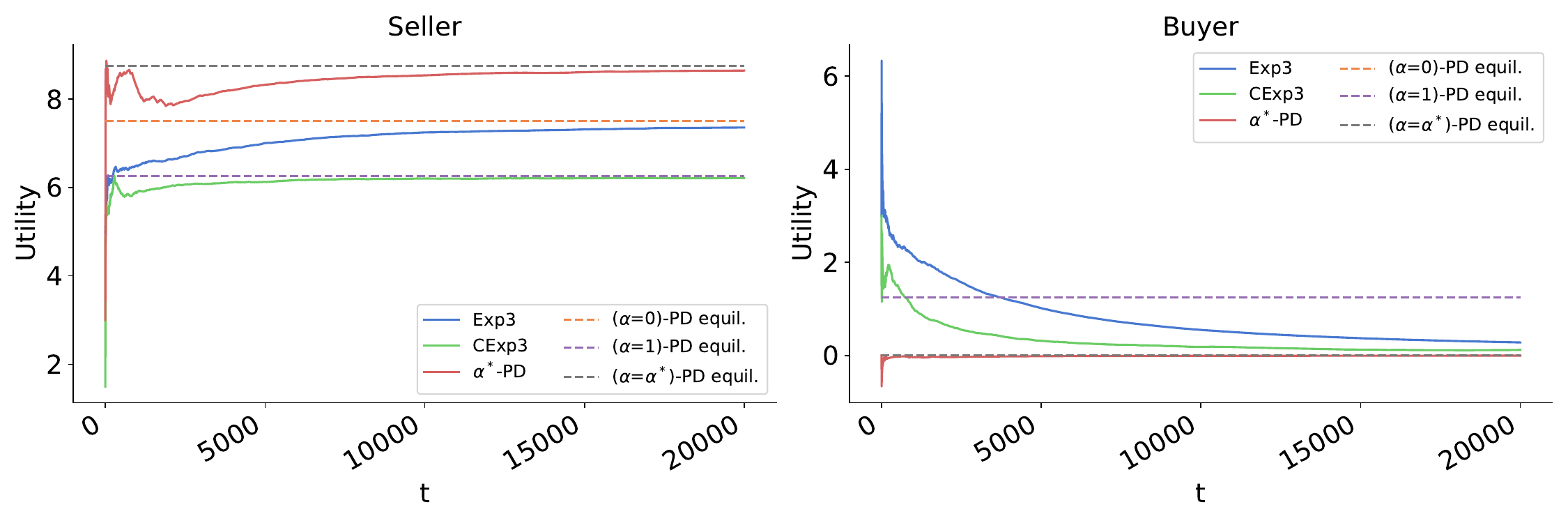}
  \caption{Convergence of seller and buyer utilities for various algorithms. $\thetalow < \mu\thetahigh$ with our experiment parameters, so the buyer's $(\alpha=0)$-PD and $(\alpha=\alpha^*)$-PD utilities are the same (see Corollary \ref{cor:order of utilities}).}
  \label{fig:convergence of utilities}
\end{figure}

\subparagraph{Consistency of $\hat{\alpha}$.}
Figure \ref{fig:consistency of alpha hat} illustrates the consistency of the buyer's estimator (\eqref{eq:consistent-estimator}). Our simulations show that the buyer's estimate $\hat{\alpha}_t$ of the seller's probability of price discrimination converges to $0$ against a seller playing Exp3, to $0.5$ against a seller playing $\alpha^*$-PD (where $\alpha^*=\nicefrac{c_B}{\Delta\theta}=0.5$ given our simulation parameters), and to higher-than-$0.5$ against a seller playing CExp3. Importantly, $\hat{\alpha}_t$ aligns with the seller's true average probability of price-discrimination, $\bar{\alpha}_t$, giving empirical evidence for Lemma \ref{lemma:consistency_implication}.

\begin{figure}[!t]
  \centering
  \includegraphics[scale=0.5]{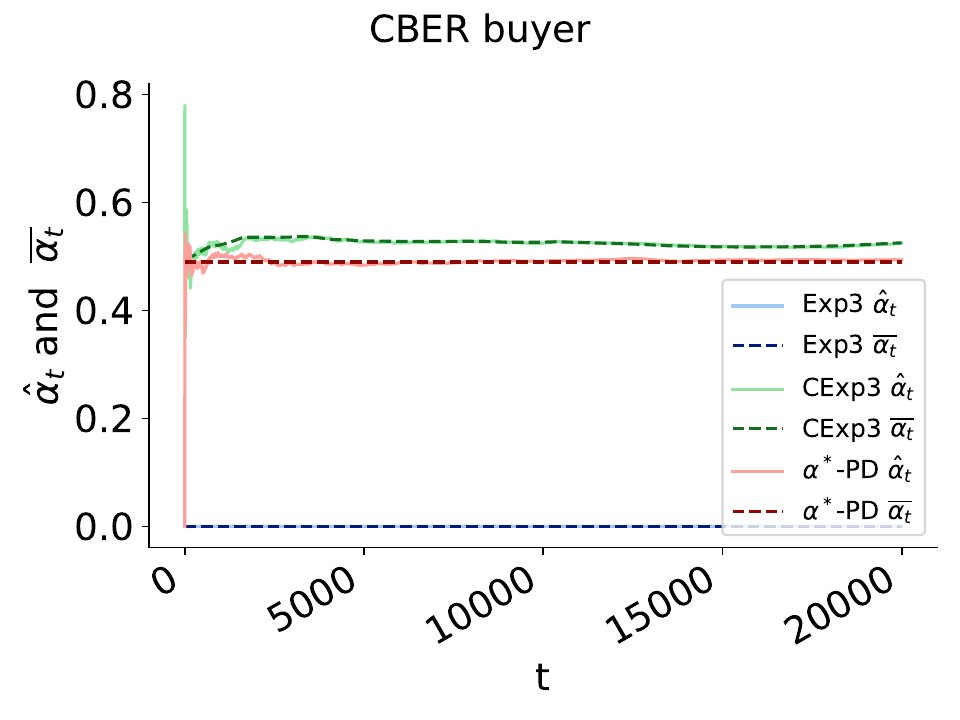}
  \caption{$\alphahat_t$ and $\bar{\alphahat}_t$ over time when seller is playing Exp3, CExp3, or $\alpha^*$-PD. In all cases, $\alphahat$ is a consistent estimator of the seller's true probability of price discrimination.}
  \label{fig:consistency of alpha hat}
\end{figure}

\subparagraph{Convergence of Seller Actions.}
In Figure \ref{fig:convergence of seller actions}, we track the cumulative proportion of the seller's price-discriminatory vs. non-price-discriminatory actions. Specifically, we track four seller actions: 1) charging a high price regardless of signal, 2) charging a low price regardless of signal, 3) charging a high price for a low signal and low price for a low signal (PD), and 4) charging a low price for a high signal and a high price for a low signal (reversePD). Given our parameter values for these simulations (i.e. $\thetalow < \mu\thetahigh$ and $\alpha^*=0.5$), in equilibrium we would expect that, for each batch of $n$ buyers at a single round: 1) a signal-blind seller sets a high price for all $n$ buyers, 2) a signal-aware seller sets a high price for high-signal buyers and a low price for low-signal buyers, and 3) a $\alpha^*$-PD seller sets a high price for all high-signal buyers and low price for all low-signal buyers with probability $0.5$ and sets a high price for all $n$ buyers with probability $0.5$. Figure \ref{fig:convergence of seller actions} gives empirical evidence for this intuition.

\begin{figure}[!t]
  \centering
  \includegraphics[scale=0.4]{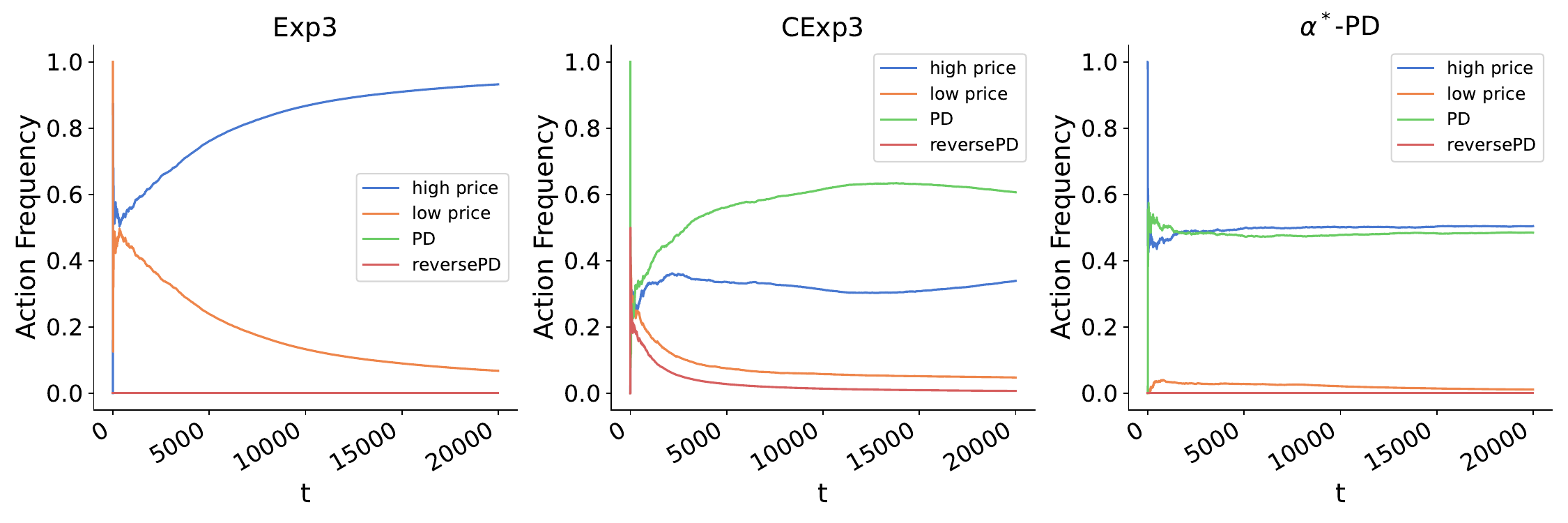}
  \caption{Relative frequency of actions for the seller playing Exp3, CExp3 and $\alpha^*$-PD. The number of PD and reversePD actions for the Exp3 seller are both $0$, as is expected.}
  \label{fig:convergence of seller actions}
\end{figure}

\subparagraph{Biased $\hat{\alpha}$.}
In realistic settings, the buyer may not have a consistent estimate of price discrimination and instead only have access to a noisy $\hat{\alpha}$. Figure \ref{fig:noisy alpha_hats} examines whether a seller can benefit from non-consistency in the buyer's estimate.  the $y$-axis of the figure tracks the seller's cumulative average utility after $20,000$ rounds of interaction with \cber-buyers. We partition the interval $[-1,1]$ into twenty segments $\gamma_i$ of width $0.1$, and the buyers use estimator $\alphahat_t + \epsilon_t$, where $\epsilon_t \sim U(\gamma_i)$. The plot then tracks the seller's cumulative average utility after $20,000$ rounds of interaction with buyers for each noise interval $\gamma_i$. If $\hat{\alpha}_t + \epsilon_t$ is less than $0$ or greater than $1$, we clip it at those values respectively. In all cases, the seller is hurt by a $\thetahigh$-buyer who overestimates the probability of price discrimination (high values of $\epsilon_t$) and is thus more likely to evade, costing the seller the evasion cost. Against a buyer who underestimates the probability of price discrimination (low values of $\epsilon_t$), neither the Exp3 nor $\alpha^*$-PD seller gains utility, since the equilibrium behavior of the buyer with consistent $\hat{\alpha}_t$ aligns with the no-price-discrimination equilibrium (see Figure \ref{fig:convergence of utilities}). By contrast, the CExp3 seller benefits from a buyer who underestimates the probability of price discrimination, since the seller benefits from discriminatory pricing without incurring the evasion cost. Against a \cber-buyer with consistent estimates, this advantage is impossible at equilibrium.

\begin{figure}[!t]
  \centering
  \includegraphics[scale=0.5]{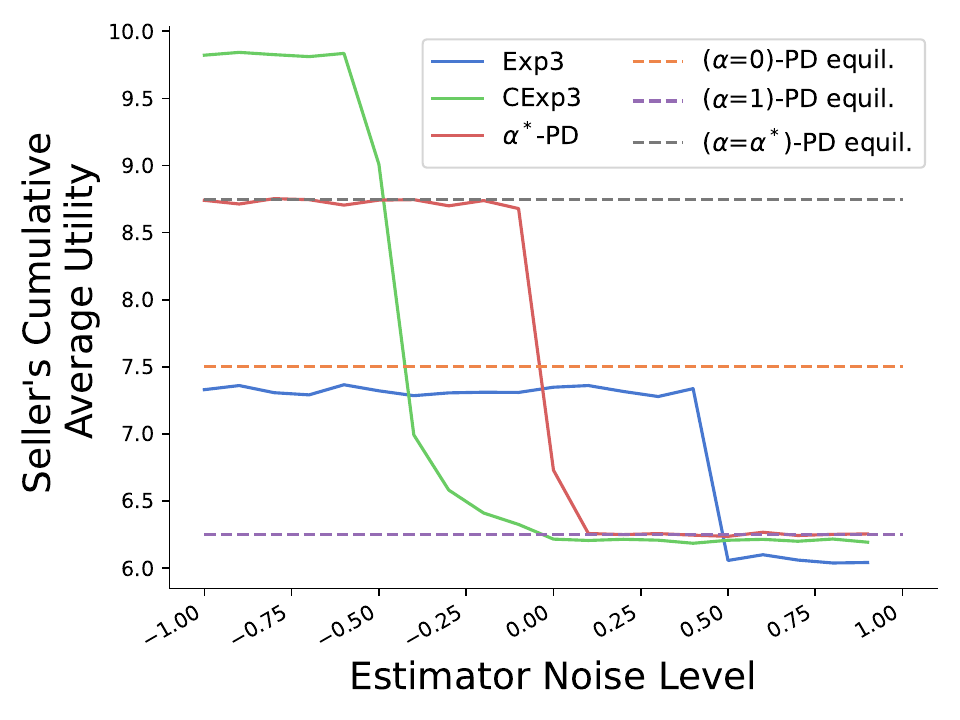}
  \caption{Cumulative average utility of the seller playing against \cber-buyers using biased $\hat{\alpha}$'s.}
  \label{fig:noisy alpha_hats}
\end{figure}
\section{Conclusion} 
Since the type and level of privacy desired generally depends on the utilities of stakeholders and forms of interaction among them, we propose a game theoretic framework for privacy in this paper. 
We analyzed the perfect Bayes Nash equilibrium in a single-interaction setting as well as no-regret and no-policy-regret dynamics emerging over repeated interactions.
In both these settings, we show how the different components of the game---utilities, actions and information sets (information available to players when choosing actions) impact the privacy levels that emerge. 

Our results shed light on the impacts of different privacy-related interventions---we showed that enabling a seller to credibly commit to privacy (e.g., through privacy legislation like the GDPR) or revealing the seller's past behavior (e.g., through privacy auditing) can surprisingly improve their utility. Thus, we believe our framework can be used to help analyze and craft privacy policies.

\section*{Acknowledgements} 

We thank Alireza Fallah and Stephen Bates for helpful discussions. TD acknowledges support from the National Science Foundation Graduate Research Fellowship Program under grant no.\ 2146752, SPK is partially supported by a Mobility Fellowship by the Swiss National Science Foundation. We also acknowledge support from the European Union (ERC-2022-SYG-OCEAN-101071601).

\bibliographystyle{plainnat}
\bibliography{bibliography}

\begin{thebibliography}{23}
\providecommand{\natexlab}[1]{#1}
\providecommand{\url}[1]{\texttt{#1}}
\expandafter\ifx\csname urlstyle\endcsname\relax
  \providecommand{\doi}[1]{doi: #1}\else
  \providecommand{\doi}{doi: \begingroup \urlstyle{rm}\Url}\fi

\bibitem[Acquisti and Varian(2004)]{acquisti2004}
Alessandro Acquisti and Hal Varian.
\newblock Conditioning prices on purchase history.
\newblock \emph{Marketing Science}, 2004.

\bibitem[Acquisti et~al.(2016)Acquisti, Taylor, and
  Wagman]{acquisti2016economics}
Alessandro Acquisti, Curtis Taylor, and Liad Wagman.
\newblock The economics of privacy.
\newblock \emph{Journal of Economic Literature}, 54\penalty0 (2):\penalty0
  442--492, 2016.

\bibitem[Arora et~al.(2020)Arora, Dinitz, Marinov, and Mohri]{arora2020policy}
Raman Arora, Michael Dinitz, Teodor~V. Marinov, and Mehryar Mohri.
\newblock Policy regret in repeated games.
\newblock \emph{Advances in Neural Information Processing Systems}, 2020.

\bibitem[Camara et~al.(2020)Camara, Hartline, and Johnsen]{camara2020prior}
Modibo~K. Camara, Jason~D. Hartline, and Aleck Johnsen.
\newblock Mechanisms for a no-regret agent: Beyond the common prior.
\newblock \emph{2020 IEEE 61st Annual Symposium on Foundations of Computer
  Science (FOCS)}, 2020.

\bibitem[Conitzer et~al.(2012)Conitzer, Taylor, and Wagman]{conitzer2012}
Vincent Conitzer, Curtis Taylor, and Liad Wagman.
\newblock Hide and seek: Costly consumer privacy in a market with repeated
  purchases.
\newblock \emph{Marketing Science}, 31\penalty0 (2):\penalty0 277--292, 2012.

\bibitem[Deng et~al.(2019)Deng, Schneider, and Sivan]{deng2019strategizing}
Yuan Deng, Jon Schneider, and Balasubramanian Sivan.
\newblock Strategizing against no-regret learners.
\newblock \emph{Advances in Neural Information Processing Systems}, 32, 2019.

\bibitem[Dong et~al.(2019)Dong, Roth, and Su]{dong2019gaussian}
Jinshuo Dong, Aaron Roth, and Weijie~J Su.
\newblock Gaussian differential privacy.
\newblock \emph{arXiv preprint arXiv:1905.02383}, 2019.

\bibitem[Dwork and Roth(2014)]{dwork2014privacy}
Cynthia Dwork and Aaron Roth.
\newblock The algorithmic foundations of differential privacy.
\newblock \emph{Foundations and Trends in Theoretical Computer Science},
  9\penalty0 (3--4):\penalty0 211--407, 2014.

\bibitem[Ely and Valimaki(2003)]{ely2003reputation}
Jeffrey Ely and Juuso Valimaki.
\newblock Bad reputation.
\newblock \emph{The Quarterly Journal of Economics}, 118(3):\penalty0 785--814,
  2003.

\bibitem[Foster and Vohra(1997)]{foster1997calibrated}
Dean~P Foster and Rakesh~V Vohra.
\newblock Calibrated learning and correlated equilibrium.
\newblock \emph{Games and Economic Behavior}, 21\penalty0 (1-2):\penalty0 40,
  1997.

\bibitem[Fudenberg and Villas-Boas(2006)]{fudenberg2006behavior}
Drew Fudenberg and J~Miguel Villas-Boas.
\newblock Behavior-based price discrimination and customer recognition.
\newblock \emph{Handbook on Economics and Information Systems}, 1:\penalty0
  377--436, 2006.

\bibitem[Haghtalab et~al.(2024)Haghtalab, Podimata, and
  Yang]{haghtalab2024calibrated}
Nika Haghtalab, Chara Podimata, and Kunhe Yang.
\newblock Calibrated {S}tackelberg games: Learning optimal commitments against
  calibrated agents.
\newblock \emph{Advances in Neural Information Processing Systems}, 36, 2024.

\bibitem[Hart and Tirole(1988)]{hart1988contract}
Oliver~D Hart and Jean Tirole.
\newblock Contract renegotiation and coasian dynamics.
\newblock \emph{The Review of Economic Studies}, 55\penalty0 (4):\penalty0
  509--540, 1988.

\bibitem[Horner(2002)]{horner2002reputation}
Johannes Horner.
\newblock Reputation and competition.
\newblock \emph{American Economic Review}, 92(3):\penalty0 644--663, 2002.

\bibitem[Ichihashi(2020)]{ichihashi2020online}
Shota Ichihashi.
\newblock Online privacy and information disclosure by consumers.
\newblock \emph{American Economic Review}, 110\penalty0 (2):\penalty0 569--595,
  2020.

\bibitem[Lattimore and Szepesvari(2020)]{lattimorebandits}
Tor Lattimore and Csaba Szepesvari.
\newblock \emph{Bandit Algorithms}.
\newblock Cambridge University Press, 2020.

\bibitem[Ligett and Nissim(2023)]{ligett2023we}
Katrina Ligett and Kobbi Nissim.
\newblock We need to focus on how our data is used, not just how it is shared.
\newblock \emph{Communications of the ACM}, 66\penalty0 (9):\penalty0 32--34,
  2023.

\bibitem[Mironov(2017)]{mironov2017renyi}
Ilya Mironov.
\newblock R{\'e}nyi differential privacy.
\newblock In \emph{2017 IEEE 30th Computer Security Foundations Symposium
  (CSF)}, pages 263--275. IEEE, 2017.

\bibitem[Montes et~al.(2015)Montes, Sand-Zantman, and Valletti]{montes2015}
Rodrigo Montes, Wilfried Sand-Zantman, and Tommaso Valletti.
\newblock The value of personal information in markets with endogenous privacy.
\newblock \emph{Center for Economic and International Studies}, 13\penalty0
  (352), 2015.

\bibitem[Nissim et~al.(2017)Nissim, Bembenek, Wood, Bun, Gaboardi, Gasser,
  O'Brien, Steinke, and Vadhan]{nissim2017bridging}
Kobbi Nissim, Aaron Bembenek, Alexandra Wood, Mark Bun, Marco Gaboardi, Urs
  Gasser, David~R O'Brien, Thomas Steinke, and Salil Vadhan.
\newblock Bridging the gap between computer science and legal approaches to
  privacy.
\newblock \emph{Harvard Journal of Law \& Technology}, 31:\penalty0 687, 2017.

\bibitem[Shapiro(1983)]{shapiro1983reputation}
Carl Shapiro.
\newblock Premiums for high quality products are returns to reputation.
\newblock \emph{Quarterly Journal of Economics}, 98(4):\penalty0 659--679,
  1983.

\bibitem[Solow-Niederman(2022)]{solow2022information}
Alicia Solow-Niederman.
\newblock Information privacy and the inference economy.
\newblock \emph{Northwestern University Law Review}, 117:\penalty0 357, 2022.

\bibitem[Tang et~al.(2017)Tang, Korolova, Bai, Wang, and Wang]{tang2017privacy}
Jun Tang, Aleksandra Korolova, Xiaolong Bai, Xueqiang Wang, and Xiaofeng Wang.
\newblock Privacy loss in {A}pple's implementation of differential privacy on
  mac{OS} 10.12.
\newblock \emph{arXiv preprint arXiv:1709.02753}, 2017.

\end{thebibliography}

\newpage
\appendix

\section{Proofs from Section~\ref{sec:single-interaction}}

\subsection{Proof of Theorem \ref{thm:single stage equilibrium}} \label{sec:proof of single stage equilibrium}

\begin{proof}
Part (a) comes from the fact that $\thetalow$ buyers have no reason to pretend to have a higher valuation for the good than they actually do. 
Part (e) comes from the fact that buyers are utility maximizing. 

Part (c) comes from the following reasoning: since \signalblind sellers cannot see the buyers' signals, they must choose one price to set for all buyers. The seller wants to maximize their revenue, so they would ideally want to set the highest price that the buyer is willing to pay ($\thetahigh$ for $\thetahigh$-buyers and $\thetalow$ for $\thetalow$-buyers). 
However, the seller does not know the type of the buyer; all they know is the probability $\mu$ that the buyer is $\thetahigh$. The seller has to make a decision between charging $\thetahigh$ or $\thetalow$. 
If the seller charges $\thetalow$, both $\thetalow$ and $\thetahigh$ agents would be willing to buy, so the expected revenue is $\thetalow$.
If the seller charges the higher price $\thetahigh$, only $\thetahigh$ agents would be willing to buy, so the expected revenue is $\mu\thetahigh$, which corresponds to
\begin{align*}
    p^*_{\mathrm{\signalblind}} = \begin{cases}
        \thetalow & \text{if } \thetalow \leq \mu \thetahigh\\
        \thetahigh & \text{if } \thetalow > \mu \thetahigh.
    \end{cases}
\end{align*}

Part (c) and (d) come from the following best-response arguments.
Our goal is to show $p^*_{\mathrm{\signalaware}}$ is a best response given $q^*$ and vice versa, where
\begin{align}
    p^*_{\mathrm{\signalaware}}(s)=
    \begin{cases}
        \thetalow & \text{if } s=\slow \text{ is observed}\\
        \thetahigh & \text{if signal } s=\shigh \text{ is observed}.
    \end{cases}
\end{align}
and 
\begin{align}
    q^*=\min\bigg\{1,\frac{(1-\mu)\thetalow}{\mu\Delta\theta}\bigg\}
\end{align}

\emph{What is the \signalaware seller's best response after seeing $\shigh$?}
From part (a), we know that $\thetalow$ buyers never signal $\thetahigh$, so the seller knows that a $\shigh$ signal implies that the buyer is type $\thetahigh$ and should therefore set a price of $\thetahigh$ after seeing $\shigh$, i.e., $p^*_{\mathrm{\signalaware}}(\shigh) = \thetahigh$. 

\emph{What is the \signalaware seller's best response after seeing $\slow$?} In order for $p^*_{\mathrm{\signalaware}}$ to be a best response, it must maximize the seller's expected utility, where the expectation is over the seller's posterior belief over the buyer's type given that they have signaled $\slow$. Given probability $q^*$ that the $\thetahigh$ buyer sends signal $\slow$, the seller's posterior belief $\hat{\mu}$ that the buyer is type $\thetahigh$ is
\begin{align}\label{eq:updated mu}
    \hat{\mu}
    =
    \mathbbm{P}(\theta=\thetahigh|s=\slow)
    =
    \frac{\mathbbm{P}(s=\slow|\theta=\thetahigh)\mathbbm{P}(\theta=\thetahigh)}{\mathbbm{P}(s=\slow|\theta=\thetahigh)\mathbbm{P}(\theta=\thetahigh) + \mathbbm{P}(s=\slow|\theta=\thetalow)\mathbbm{P}(\theta=\thetalow)}
    =
    \frac{q^*\mu}{q^*\mu + 1-\mu}.
\end{align}
Let $f(p)$ denote the seller's expected utility from charging price $p$ after observing signal $\slow$, so
\begin{align}\label{eq:signalaware seller utility}
    f(p)
    =
    \begin{cases}
        p - \hat{\mu} q^* c_S  & \text{if } p < \thetalow\\
        \hat{\mu}p - \hat{\mu} q^* c_S & \text{if } p \in [\thetalow,\thetahigh].
    \end{cases}
\end{align}
In order for $p^*_{\mathrm{\signalaware}}(\slow)$ to be a best response, it must be the value that maximizes $f$: 
\begin{align}\label{eq:signalaware equilibrium price}
    p^*_{\mathrm{\signalaware}}(\slow)
    = \max_p f(p) 
    = 
    \begin{cases}
        \thetalow & \text{if } q^* \leq \min\bigg\{1,\frac{(1-\mu)\thetalow}{\mu\Delta\theta}\bigg\}\\
        \thetahigh & \text{else}.
    \end{cases} 
    = \thetalow,
\end{align}
where the last equality comes from the choice of $q^*$. This shows that $p^*_{\mathrm{\signalaware}}(\slow)=\thetalow$ is a best response for the seller.  
We now turn our attention to the $\thetahigh$-buyer. 

\emph{What is the optimal probability $q^*$ of evasion for the $\thetahigh$-buyer?} 
Let $g(q)$ denote the 
the expected utility for the $\thetahigh$ buyer when they evade with probability $q$, given that the seller is playing $p^*_{\signalblind}$ if they are \signalblind and $p^*_{\signalaware}$ if they are \signalaware, so 
\begin{align}
    g(q) 
    &= 
    \mathbbm{P}(\text{seller is \signalblind})(\text{$\thetahigh$-buyer utility if seller plays } p^*_{\mathrm{\signalblind}})\\
    &\phantom{{}=1} + 
    \mathbbm{P}(\text{seller is \signalaware})(\text{$\thetahigh$-buyer utility if seller plays } p^*_{\mathrm{\signalaware}})\\
    &=
    (1-\alpha)(\text{$\thetahigh$-buyer utility if seller plays } p^*_{\mathrm{\signalblind}}) + \alpha(\text{$\thetahigh$-buyer utility if seller plays } p^*_{\mathrm{\signalaware}})\\
    &=
    (1-\alpha)[(\Delta\theta-c_Bq)\mathbbm{1}(\thetalow \geq \mu\thetahigh)+(-c_Bq)\mathbbm{1}(\thetalow < \mu\thetahigh)]\\
    &\phantom{{}=1} +
    \alpha[(\Delta\theta-c_B)q\mathbbm{1}(q \leq \min\{1,\nicefrac{(1-\mu)\thetalow}{\mu\Delta\theta}\})+(-c_Bq)\mathbbm{1}(q>\min\{1,\nicefrac{(1-\mu)\thetalow}{\mu\Delta\theta}\})]\label{eq:buyer equilibrium utility}.
\end{align}
We analyze \eqref{eq:buyer equilibrium utility} in cases.
\begin{itemize}
    \item If $1\leq\nicefrac{(1-\mu)\thetalow}{\mu\Delta\theta}$, this implies that $\thetalow \geq \mu\thetahigh$, so \eqref{eq:buyer equilibrium utility} simplifies to
    \begin{align}
        u_B
        &=
        (1-\alpha)\Delta\theta + (\alpha\Delta\theta-c_B)q.
    \end{align}
    \item If $\nicefrac{(1-\mu)\thetalow}{\mu\Delta\theta} \leq 1$, this implies $\theta < \mu\thetahigh$, so \eqref{eq:buyer equilibrium utility} simplifies to
    \begin{align}
        u_B = 
        \begin{cases}
            (\alpha\Delta\theta-c_B)q &\text{if } q \leq \nicefrac{(1-\mu)\thetalow}{\mu\Delta\theta}\\
            -c_Bq &\text{else}.
        \end{cases}
    \end{align}
\end{itemize} 
Combining both cases, we see that the $\thetahigh$-buyer's optimal probability of evasion is
\begin{align}
    q^*=
    \begin{cases}
        \min\bigg\{1,\frac{(1-\mu)\thetalow}{\mu\Delta\theta}\bigg\} & \text{if } \alpha > \nicefrac{c_B}{\Delta\theta}\\
        0 & \text{else}.
    \end{cases}
\end{align}
\end{proof}

\subsection{Proof of Corollary \ref{cor:order of utilities}} \label{sec: proof of order of utilities}

\begin{proof}
We summarize fundamental properties of the equilibrium in Theorem \ref{thm:single stage equilibrium}, from which expressions for the buyer's and seller's expected utilities follow. At equilibrium,
\begin{enumerate}
    \item When $\thetalow \geq \mu\thetahigh$, the \signalblind seller always sets price $\thetalow$.
    \item When $\thetalow < \mu\thetahigh$, the \signalblind seller always sets price $\thetahigh$.
    \item When $\thetalow \geq \mu\thetahigh$ and $\alpha > \nicefrac{c_B}{\Delta\theta}$, the $\thetahigh$-buyer always evades (since $1 \leq \nicefrac{(1-\mu)\thetalow}{\mu\Delta\theta}$ when $\thetalow \geq \mu\thetahigh$).
    \item When $\thetalow < \mu\thetahigh$ and $\alpha > \nicefrac{c_B}{\Delta\theta}$, the $\thetahigh$-buyer evades with probability $\nicefrac{(1-\mu)\thetalow}{\mu\Delta\theta}$ (since $\nicefrac{(1-\mu)\thetalow}{\mu\Delta\theta} < 1$ when $\thetalow < \mu\thetahigh$).
    \item The $\thetahigh$-buyer always signals truthfully when $\alpha \leq \nicefrac{c_B}{\Delta\theta}$.
    \item The $\thetalow$-buyer always signals truthfully.
\end{enumerate}
\subparagraph{Buyer Utilities.}
It is straightforward to see that the $\thetalow$-buyer's expected utility is zero, so we focus on the $\thetahigh$-buyer. 
The $\thetahigh$-buyer's expected utility is
\begin{align}
    u_B
    &= 
    \mathbbm{P}(\text{seller is \signalblind})[u_B|\text{seller is \signalblind}]\\
    &\phantom{{}=1} + 
    \mathbbm{P}(\text{seller is \signalaware})[u_B|\text{seller is \signalaware}]\\
    &=
    (1-\alpha)[u_B|\text{seller is \signalblind}] + \alpha[u_B|\text{seller is \signalaware}]\label{eq:buyer utility}.
\end{align}
where, by the properties above,
\begin{itemize}
    \item If $\thetalow \geq \mu\thetahigh$, 
        \begin{align}
        u_B|\text{seller is \signalblind}
        &= 
        (\Delta\theta - c_B)\mathbbm{1}(\alpha>\nicefrac{c_B}{\Delta\theta}) + \Delta\theta\mathbbm{1}(\alpha\leq\nicefrac{c_B}{\Delta\theta}). \\
        u_B|\text{seller is \signalaware}
        &=
        (\Delta\theta-c_B)\mathbbm{1}(\alpha \leq \nicefrac{c_B}{\Delta\theta}).
        \end{align}
    \item If $\thetalow < \mu\thetahigh$,
    \begin{align}
    u_B|\text{seller is signal blind}
    &=
    -c_B(\nicefrac{(1-\mu)\thetalow}{\mu\Delta\theta})\mathbbm{1}(\alpha>\nicefrac{c_B}{\Delta\theta}) \\
    u_B|\text{seller is signal aware}
    &=
    (\Delta\theta-c_B)(\nicefrac{(1-\mu)\thetalow}{\mu\Delta\theta})\mathbbm{1}(\alpha \leq \nicefrac{c_B}{\Delta\theta})
\end{align}
\end{itemize}

\subparagraph{Seller Utility.} The seller's expected utility is
\begin{align}
    u_S
    &= 
    \mathbbm{P}(\text{seller is \signalblind})[u_S|\text{seller is \signalblind}]\\
    &\phantom{{}=1} + 
    \mathbbm{P}(\text{seller is \signalaware})[u_S|\text{seller is \signalaware}]\\
    &=
    (1-\alpha)[u_S|\text{seller is \signalblind}] + \alpha[u_S|\text{seller is \signalaware}]\label{eq:seller utility}.
\end{align}
where, by the properties above,
\begin{itemize}
    \item If $\thetalow \geq \mu\thetahigh$,
        \begin{align}
            u_S|\text{seller is \signalblind}
            &=
            \mathbbm{P}(\text{buyer is type-}\thetalow)(u_S|\text{seller is \signalblind and buyer is type-$\thetalow$})\\
            &\phantom{{}=1} +
            \mathbbm{P}(\text{buyer is type-}\thetahigh)(u_S|\text{seller is \signalblind and buyer is type-$\thetahigh$})\\
            &=
            (1-\mu)\thetalow + \mu[(\thetalow-c_S)\mathbbm{1}(\alpha > \nicefrac{c_B}{\Delta\theta}) + \thetalow\mathbbm{1}(\alpha \leq \nicefrac{c_B}{\Delta\theta})]
        \end{align}
    and
        \begin{align}
            u_S|\text{seller is \signalaware}
            &=
            \mathbbm{P}(\text{buyer is type-}\thetalow)(u_S|\text{seller is \signalaware and buyer is type-$\thetalow$})\\
            &\phantom{{}=1}+ 
            \mathbbm{P}(\text{buyer is type-}\thetahigh)(u_S|\text{seller is \signalaware and buyer is type-$\thetahigh$})\\
            &=
            (1-\mu)\thetalow + \mu[(\thetalow-c_S)\mathbbm{1}(\alpha > \nicefrac{c_B}{\Delta\theta}) + \thetahigh\mathbbm{1}(\alpha \leq \nicefrac{c_B}{\Delta\theta})].
        \end{align}
    \item If $\thetalow < \mu\thetahigh$,
        \begin{align}
            u_S|\text{seller is \signalblind}
            &=
            \mathbbm{P}(\text{buyer is type-}\thetalow)(u_S|\text{seller is \signalblind and buyer is type-$\thetalow$})\\
            &\phantom{{}=1} +
            \mathbbm{P}(\text{buyer is type-}\thetahigh)(u_S|\text{seller is \signalblind and buyer is type-$\thetahigh$})\\
            &=
            (1-\mu)0\\
            &\phantom{{}=1}+
            \mu[(\thetahigh-c_S)(\nicefrac{(1-\mu)\thetalow}{\mu\Delta\theta})\mathbbm{1}(\alpha>\nicefrac{c_B}{\Delta\theta})\\
            &\phantom{{}=1} +
            \thetahigh(1-(\nicefrac{(1-\mu)\thetalow}{\mu\Delta\theta}))\mathbbm{1}(\alpha>\nicefrac{c_B}{\Delta\theta})+\thetahigh\mathbbm{1}(\alpha\leq\nicefrac{c_B}{\Delta\theta})]
        \end{align}
        and 
        \begin{align}
            u_S|\text{seller is \signalaware}
            &=
            \mathbbm{P}(\text{buyer is type-}\thetalow)(u_S|\text{seller is \signalaware and buyer is type-$\thetalow$})\\
            &\phantom{{}=1}+ 
            \mathbbm{P}(\text{buyer is type-}\thetahigh)(u_S|\text{seller is \signalaware and buyer is type-$\thetahigh$})\\
            &=
            (1-\mu)\thetalow\\
            &\phantom{{}=1}+
            \mu[(\thetalow-c_S)(\nicefrac{(1-\mu)\thetalow}{\mu\Delta\theta})\mathbbm{1}(\alpha>\nicefrac{c_B}{\Delta\theta})\\
            &\phantom{{}=1}+
            \thetahigh(1-(\nicefrac{(1-\mu)\thetalow}{\mu\Delta\theta}))\mathbbm{1}(\alpha>\nicefrac{c_B}{\Delta\theta})+\thetahigh\mathbbm{1}(\alpha\leq\nicefrac{c_B}{\Delta\theta})].
        \end{align}
\end{itemize}

Plugging in the relevant values of $\alpha$ to the buyer's \eqref{eq:buyer utility} and seller's \eqref{eq:seller utility} utility expressions gives the stated orderings. We also see that $\alpha^*=\nicefrac{c_B}{\Delta\theta}$ maximizes the seller's utility.
\end{proof}

\section{Repeated PD Protocol} \label{sec:repeated_PD_protocol}

The following algorithm makes clear what information is available to a given player at each point in the repeated interaction setting. 

\begin{algorithm}
\caption{Repeated PD protocol}
\label{alg:repeated_pd}
\begin{algorithmic}[1] 
\STATE Parameters: $(\mu, \thetalow, \thetahigh, c_B, c_S)$
\STATE ${H}^B_1, {H}^S_1 := \varnothing$
\FOR{$t = 1$ to $T$} 
    \FOR{each buyer $i = 1$ to $N$}
        \STATE Nature draws $\theta_i^t$ with $\theta_i^t = \thetahigh$ with probability $\mu$ and $\theta_i^t = \thetalow$ otherwise
        \STATE Buyer $i$ chooses mixed strategy $\pi_S^i$ over signal $s_i^t \in \{\slow, \shigh\}$ based on $\{\theta_i^t\} \cup H^B_t$ 
        \STATE Seller chooses mixed strategy $\pi^p(\cdot \mid s=s_i^t)$ over price $p_i^t$.
        \STATE Buyer $i$ decides to buy, denoted by indicator $b_i^t$, based on $\{\theta_i^t, p_i^t\} \cup H^B_t$
    \ENDFOR
    \STATE Buyer $i$ receives utility $u_B(\theta_i^t, s_i^t, p_i^t, b_i^t)$ and seller receives utility $u_S \left ((\theta_i^t, s_i^t, p_i^t, b_i^t)_{i=1}^n \right )$, as defined in \eqref{eq:buyer_utility_function} and \eqref{eq:seller_utility_function}.
    \STATE $H^B_{t+1} = H^B_t \cup \{(\theta_i^t, s_i^t, p_i^t)_{i=1}^n \}$
    \STATE $H^S_{t+1} = H^S_t \cup \{(s_i^t, p_i^t)_{i=1}^n \}$
\ENDFOR
\end{algorithmic}
\end{algorithm}

\section{Proofs from Section~\ref{sec:repeated_interaction}}

\subsection{Proof of Proposition \ref{prop:existence_of_consistent_estimator}} \label{sec:proof of existence of consistent estimator}

\begin{proof}
    For each round $t$, let 
    \begin{align*}
        I_t = \ind{\exists i  \text{ s.t. } s_i^t = \shigh \text{ and } \exists j  \text{ s.t. } s_j^t = \slow }
    \end{align*}
    be an indicator for whether both types of signals are observed at round $t$, i.e., whether round $t$ is ``informative" about if there is price discrimination. For rounds $t$ with $I_t=1$, we additionally
    define the following random variables: 
    \begin{itemize}
        \item $\overline{P}_t = p_i^t$ for the smallest $i \in [N]$ such that $s_i^t=\shigh$,
        \item $\underline{P}_t = p_j^t$ for the smallest $j \in [N]$ such that $s_j^t=\slow$, and
        \item $X_t = \ind{\overline{P}_t \neq \underline{P}_t}$, an indicator for observed price discrimination.   
    \end{itemize}
    Note that the choice to define $\overline{P}_t$ and $\underline{P}_t$ to correspond to the \emph{smallest} index satisfying the corresponding condition is simply for concreteness; we could equivalently sample uniformly from the set of indices satisfying the condition.
    
    Recall that $H_t= ((\theta_i^{\tau}, s_i^{\tau}, p_i^{\tau})_{i=1}^n)_{\tau=1}^{t-1}$ is the history known by buyers at the beginning of round $t$.
    Consider the following estimator:
    \begin{align}\label{eq:consistent-estimator}
        \alphahat_t = \frac{1}{t}\sum_{\tau=1}^{t-1} \frac{X_{\tau} I_{\tau}}{\E[I_{\tau}|H_{\tau}]}
    \end{align}
    The expectation $\E[I_{\tau}|H_{\tau}]$ is over the randomness at round $\tau$.
    Note that $\alphahat_t$ is computable based on the history $H_t$, because $\E[I_{\tau}|H_{\tau}]$ is computable for any $\tau < t$. 
    We will now show that $\alphahat_t$ satisfies Definition \ref{def:consistentbeliefs}. We start by computing the expectation of $\alphahat_t$:
    \begin{align*}
        \E[\alphahat_t] &= \E\left[\frac{1}{t}\sum_{\tau=1}^{t-1}\frac{X_{\tau}I_{\tau}}{\E[I_{\tau}|H_{\tau}]}\right] \\
        &= \frac{1}{t} \sum_{\tau=1}^{t-1}\E\left[\frac{X_{\tau}I_{\tau}}{\E[I_{\tau}|H_{\tau}]}\right]    & \text{linearity of expectation} \\
        &= \frac{1}{t} \sum_{\tau=1}^{t-1}\E\left[\E\left[\frac{X_{\tau}I_{\tau}}{\E[I_{\tau}|H_{\tau}]} \bigg | H_{\tau}\right]\right]    & \text{tower rule}   \\
        &= \frac{1}{t} \sum_{\tau=1}^{t-1}\E\left[\frac{\E[X_{\tau}I_{\tau}|H_{\tau}]}{\E[I_{\tau}|H_{\tau}]}\right]     \\
        \intertext{Observe that $X_{\tau}$ and $I_{\tau}$ are independent given $H_{\tau}$. To see why, note that the randomness in $X_{\tau} | H_{\tau}$ comes only from the randomness in the seller's mixed strategy at round $\tau$, whereas the randomness in $I_{\tau} | H_{\tau}$ comes only from the randomness in the buyers mixed strategy at round $\tau$. The mixed strategies are fixed given $H_{\tau}$, and the additional randomness is independent. Thus,}
        &= \frac{1}{t} \sum_{\tau=1}^{t-1}\E\left[\frac{\E[X_{\tau}|H_{\tau}]\E[I_{\tau}|H_{\tau}]}{\E[I_{\tau}|H_{\tau}]}\right] \\
        &= \frac{1}{t} \sum_{\tau=1}^{t-1}\E[\E[X_{\tau}|H_{\tau}]] \\
        &= \frac{1}{t} \sum_{\tau=1}^{t-1}\E\left[\E\big[ \ind{\overline{P}_{\tau} \neq \underline{P}_{\tau}}|H_{\tau}\big]\right] & \text{by definition of $X_{\tau}$} \\
        \intertext{Since $\overline{P}_{\tau}|H_{\tau} \sim \pi^p_t(\cdot|s=\shigh)$ and $\underline{P}_{\tau}|H_{\tau} \sim \pi^p_t(\cdot|s=\slow)$ by definition of the game, we have}
        &= \frac{1}{t} \sum_{\tau=1}^{t-1} \alpha_{\tau}
    \end{align*}
    Finally, plugging in the above expression with $t=T$ into the criterion for consistency, we have
    \begin{align}
    \lim_{T \rightarrow \infty} \left \lvert \E [\alphahat_T] -  \frac 1 T \sum_{t=1}^T \alpha_t \right \rvert 
    = \lim_{T \rightarrow \infty} \left \lvert \frac{1}{T} \sum_{t=1}^{T-1} \alpha_t -  \frac{1}{T} \sum_{t=1}^T \alpha_t \right \rvert 
    = \lim_{T \rightarrow \infty} \frac{\alpha_T}{T} 
    = 0 
    \end{align}
    as desired.
    The last equality comes from the fact that $\alpha_T$ is a probability, so it is bounded between 0 and 1 for all $T$. 
\end{proof}

\subsection{Proof of Proposition~\ref{prop:regret min may not increase utility}}\label{proof:regret min may not increase utility}
\begin{proof}
First, we will show that always price-discriminating ($\pi^p_t = \pi^p_{\PD}$ for all $t \in [T]$) is no-regret against \cber-buyers. For \cber-buyers, their strategy $\pi^s_t$ at each round $t$ is either $\truthstrat$ or $\stratstrat$.
For both these buyer responses, the seller's optimal strategy is to always price discriminate as shown in the computation of the seller's equilibrium response in the proof of Theorem \ref{thm:single stage equilibrium}. 
In other words, the seller incurs zero regret in each round and thus zero average regret.

Next, we will analyze the seller's average utility. Note that when $\pi^p_t = \pdstrat$, the probability of seeing different prices for different signals is $\alpha_t=1$, so $(1/t)\sum_{\tau=1}^t \alpha_{\tau} = 1$ for all $t$. By the consistency property, $\alphahat_t$ becomes greater than $\alpha^*$ eventually (where $\alpha^*$ is as defined in Corollary \ref{cor:stackelberg}) and the buyer plays $\stratstrat$. 
In other words, eventually the seller and buyers will all be playing their equilibrium strategies for the PD-game with $\alpha=1$, so their average utilities will converge to the corresponding equilibrium utilities. 
We make this argument formal below. 

Define $\kappa < \infty$ to be the maximum utility that can be achieved by a seller in any round. The finiteness of $\kappa$ is guaranteed by definition of the seller's utility function. Define $A_T = \{\exists t > \sqrt{T} \text{ s.t. } \alphahat_t > \alpha^*\}$ and let $A_T^C=\{\alphahat_t > \alpha^* \text{ for all } t > \sqrt{T}\}$ denote the complement. Let $\gamma_T = \PP(A_T)$ and $1-\gamma_T = \PP(A_T^C)$ denote the corresponding probabilities. Then, we can decompose the expected average seller's utility as 
\begin{align} \label{eq:noregret-utility-proof}
    \E\left[\frac 1 T \sum_{t=1}^T U_S(\pi_t)\right] = \gamma_t \E\left[\frac 1 T \sum_{t=1}^T U_S(\pi_t) \Bigg| A_T \right] + (1-\gamma_t) \E\left[\frac 1 T \sum_{t=1}^T U_S(\pi_t) \Bigg| A_T^C \right].
\end{align}

The first term of \eqref{eq:noregret-utility-proof} is trivially upper bounded by $\gamma_T \kappa$.

To bound the second term of \eqref{eq:noregret-utility-proof}, first note that for any round $t$ where $\alphahat_t > \alpha^*$, the buyer's strategy will be equivalent to their equilibrium strategy with $\alpha=1$. Thus, the best utility that the seller can achieve for those rounds is $u_S(1)$. It follows that under the condition that $\alphahat_t > \alpha^*$ for every $t > \sqrt{T}$, we have 
\begin{align*}
    \frac 1 T \sum_{t=1}^T U_S(\pi_t) &= \frac 1 T \sum_{t=\sqrt{T}}^T U_S(\pi_t) + \frac 1 T \sum_{t=1}^{\sqrt{T}} U_S(\pi_t) \\
    &\leq \frac 1 T \sum_{t=\sqrt{T}}^T u_S(1) + \frac 1 T \sum_{t=1}^{\sqrt{T}} \kappa \\
    &= \frac{T -\sqrt{T}}{T} u_S(1) + \frac{\sqrt{T}\kappa}{T} \\
    &\leq u_S(1) + \frac{\kappa - u_S(1)}{\sqrt{T}}.
\end{align*}
Plugging back into \eqref{eq:noregret-utility-proof}, we get 
\begin{align} 
    \E\left[\frac 1 T \sum_{t=1}^T U_S(\pi_t)\right] \leq \gamma_t \kappa + (1-\gamma_t) \left(u_S(1) + \frac{\kappa - u_S(1)}{\sqrt{T}}\right). \\
\end{align}
By the consistency property (Lemma \ref{lemma:consistency_implication}), we know $\lim_{T\to \infty} \gamma_T = 0$, so
\begin{align} 
    \lim_{T\to \infty} \E\left[\frac 1 T \sum_{t=1}^T U_S(\pi_t)\right] \leq u_S(1). \\
\end{align}

\end{proof}

\subsection{Missing Proofs of Lemmas in Proof of Proposition~\ref{prop:regretOddsStack}}\label{proof:regretOddsStack}
\begin{proof}[Proof of Lemma~\ref{lem:lowTruthfulLowUtility}]
     Based on the utility orderings from Corollary \ref{cor:order of utilities},
    note the following ordering of seller utilities for different combinations of buyer and seller policies:
    \[U_S(\truthstrat, \pi^p_{\PD}) > U_S(\truthstrat, \pi^p_{\noPD}) > U_S(\stratstrat, \pi^p_{\PD}) > U_S(\stratstrat, \pi^p),\]
    where $\pi^p$ is any other pricing strategy besides $\pi^p_{\PD}$ and $\pi^p_{\noPD}$. We can then write
        \begin{align*}
            \frac 1 T\sum_{t=1}^T U_S(\pi_t) &\leq  \frac 1 T \sum_{t \in \mathcal{T}} U_S(\truthstrat, \nopdstrat) + \sum_{t \in [T] \backslash \mathcal{T}} U_S(\stratstrat, \pdstrat) \\
            &= \frac {|\mathcal{T}|} T U_S(\truthstrat, \nopdstrat) + \left ( 1 - \frac {|\mathcal{T}|} T \right ) U_S(\stratstrat, \pdstrat)  \\
            \lim_{T \rightarrow \infty} \frac 1 T \sum_{t=1}^T U_S(\pi_t) &\leq U_S(\truthstrat, \nopdstrat) \lim_{T \rightarrow \infty} \frac {|\mathcal{T}|} T  + U_S(\stratstrat, \pdstrat) \left ( 1 - \lim_{T \rightarrow \infty} \frac {|\mathcal{T}|} T \right ) 
            \intertext{Since $U_S(\truthstrat, \pdstrat) > U_S(\stratstrat, \pdstrat)$, the above upper bound on the limit of the average utility is increasing as $\lim_{T \rightarrow \infty}|\mathcal{T}| / T$ is increasing. When $\lim_{T \rightarrow \infty}|\mathcal{T}| / T \leq \alpha^*$, }
            &\leq \alpha^* U_S(\truthstrat, \pdstrat) + (1 - \alpha^*) U_S(\stratstrat, \pdstrat) \\
            &< \alpha^* U_S(\truthstrat, \pdstrat) + (1 - \alpha^*) U_S(\truthstrat, \pdstrat) \\
            &= \text{\stackutility}
        \end{align*}

     \end{proof}

     \begin{proof}[Proof of Lemma~\ref{lem:ImplicationNoRegret}]
    Let $R^S_T$ denote the average seller utility in the $T$ rounds. Since the seller is no regret, $\lim_{T \rightarrow \infty} R^S_T = 0$.
    
    Consider the regret due to the seller deviating to $\pdstrat$ in each round. The gain in utility in each round due to this deviation is non-negative since $\pdstrat$ is the best-response to both possible buyer strategies $\truthstrat, \stratstrat$. We can then lower bound the regret by considering regret accumulated in rounds where $\alphahat_t \leq \alpha^*$. In such rounds, all buyers are truthful, so whenever the seller does not charge a buyer the price corresponding to their signal type, they incur regret. The probability that the seller observes $\shigh$ but charges $\thetahigh$ is $\mu \pi^p_t(\thetalow|\slow)$, and this yields a loss of utility of $\Delta \theta$, because the buyer is type $\thetahigh$. Similarly, the probability that the seller observes $\slow$ but charges $\thetahigh$ is $(1-\mu)\pi^p_t(\thetahigh|\slow)$, and this yields a loss of utility of $\thetalow$, since the buyer is type $\thetalow$.
        \begin{align*}
        R^S_T &\geq \frac 1 T \sum_{t: \alphahat_t \leq \alpha^*} \mu \Delta \theta \pi^p_t(\thetalow | \shigh) + (1-\mu) \thetalow \pi^p_t(\thetahigh | \slow)
        \intertext{Let $\kappa = \min\{\mu \Delta \theta, (1-\mu) \thetalow\}$}
        &\geq \frac 1 T \sum_{t: \alphahat_t \leq \alpha^*} \kappa (1 - (\pi^p(\thetahigh | \shigh) - \pi^p(\thetahigh | \slow)))
    \end{align*}
    \begin{align*}
        \Longrightarrow \qquad \frac 1 T \sum_{t \in \mathcal{T}} \left (\pi^p_t(\thetahigh | \shigh) - \pi^p_t(\thetahigh | \slow) \right ) &\geq \frac{|\mathcal{T}|}{T} - \frac{R^S_T}{\kappa} \tag{1} \label{eq:regretImpIneq}
    \end{align*} 

        The above inequality shows that $|\mathcal{T}| / T$ is bounded above by some measure of simultaneous truthfulness and price discrimination. Each quantity $\pi^p_t(\thetahigh | \shigh) - \pi^p_t(\thetahigh | \shigh)$ is a measure of price-discrimination in each round and is related to $\alpha_t$ as described in the following lemma.

        \begin{lemma}\label{lem:relnPDMeasures}
            When seller pricing strategies are supported on $\{\thetahigh, \thetalow\}$, $\alpha_t \geq \pi^p_t(\thetahigh | \shigh) - \pi^p_t(\thetahigh | \shigh)$
        \end{lemma}
        \begin{proof}[Proof of Lemma~\ref{lem:relnPDMeasures}]
            Since seller pricing strategies are supported on $\{\thetahigh, \thetalow\}$, $\alpha_t$ which is the probability of seeing different prices for different signals is
            \begin{align*}
            \alpha_t &= \pi^p_t(\thetahigh | \shigh) \pi^p_t(\thetalow | \slow) +  \pi^p_t(\thetalow | \shigh) \pi^p_t(\thetahigh | \slow) \\
            &= \pi^p_t(\thetahigh | \shigh) (1 - \pi^p_t(\thetahigh | \slow)) + (1  - \pi^p_t(\thetahigh | \shigh))\pi^p_t(\thetahigh | \slow) \\
            &= \pi^p_t(\thetahigh | \shigh) + \pi^p_t(\thetahigh | \slow) - 2 \pi^p_t(\thetahigh | \shigh) \pi^p_t(\thetahigh | \slow) \\
            &= \pi^p_t(\thetahigh | \shigh) - \pi^p_t(\thetahigh | \slow) + 2\pi^p_t(\thetahigh | \slow)(1 - \pi^p_t(\thetahigh | \shigh)) \\
            &\geq \pi^p_t(\thetahigh | \shigh) - \pi^p_t(\thetahigh | \slow)
        \end{align*} 
        \end{proof}
        
    By inequality~\ref{eq:regretImpIneq}, Lemma~\ref{lem:relnPDMeasures}, and since $\lim_{T \rightarrow \infty} R^S_T / \kappa = 0$,
    \begin{align*}
        \lim_{T \rightarrow \infty} \frac{|\mathcal{T}|}{T} &\leq \lim_{T \rightarrow \infty} \frac 1 T \sum_{t \in \mathcal{T}} \alpha_t.
    \end{align*}
    
    \end{proof}

    \begin{proof}[Proof of Lemma~\ref{lem:consistencySimultaneousPDTruthful}]
    Consider the last index $t^*$ in $ \mathcal{T}$. Let us consider two cases.  The first case is $\lim_{T \rightarrow \infty} t^* / T < \alpha^*$. Then,
    \begin{align*}
        \sum_{t \in \mathcal{T}} \alpha_t / T &\leq |\mathcal{T}| / T \\
        &\leq t^* / T \\
        \Longrightarrow \lim_{T \rightarrow \infty}\sum_{t \in \mathcal{T}} \alpha_t / T &\leq \alpha^*.
    \end{align*}
    
    In the second case, $\lim_{T \rightarrow \infty} t^* = \infty$. Consider $\bar{\alpha}_{t^*} = \frac 1 {t^*} \sum_{t \leq t^*} \alpha_t \geq \sum_{t \in \mathcal{T}} \alpha_t / T$. By the consistency property, $\lim_{T \rightarrow \infty} \bar{\alpha}_{t^*} = \hat{\alpha}_{t^*}$. $\hat{\alpha}_{t^*} \leq \alpha^*$ since $t^* \in \mathcal{T}$. 
    \end{proof}

\subsection{Proof of Propsition \ref{prop:policyRegretStack}} 
\label{sec:proof_of_policyRegretStack}

\begin{proof}
    Under the conditions of this proposition, the seller's utility must, by definition of policy regret, approach a utility at least as high (or better) than the utility of any strategy in $\mathbb{A}_S^{MS}$ as $T \to \infty$. Recall that \stackutility is the seller utility achieved in the PD game when $\alpha = \alpha^*$. 
    For any $\epsilon > 0$, there exists a value $\tilde{\alpha} < \alpha^*$ such that the seller's utility in the PD game with $\alpha=\tilde{\alpha}$ (denoted $u_S(\tilde{\alpha})$) is at least $\text{\stackutility}- \epsilon$. 
    Consider the seller's mixed strategy of price-discriminating with probability $\tilde{\alpha}$, so $\alpha_t = \tilde{\alpha}$ for all $t$. Similar to the proof of Proposition \ref{prop:regret min may not increase utility}, define $A_T = \{\exists t > \sqrt{T} \text{ s.t. } \alphahat_t > \alpha^*\}$ and let $A_T^C=\{\alphahat_t > \alpha^* \text{ for all } t > \sqrt{T}\}$ denote the complement. Let $\gamma_T = \PP(A_T)$ and $1-\gamma_T = \PP(A_T^C)$ denote the corresponding probabilities.
    The expected average seller's utility can be decomposed as 
    $\E\left[\frac 1 T \sum_{t=1}^T U_S(\pi_t)\right] = \gamma_t \E\left[\frac 1 T \sum_{t=1}^T U_S(\pi_t) | A_T \right] + (1-\gamma_t) \E\left[\frac 1 T \sum_{t=1}^T U_S(\pi_t) | A_T^C \right].$ Define $m$ to be the smallest utility that can be achieved by a seller in any round. Then the first term on the right side is trivially lower bounded by $\gamma_T m$. 

    To bound the second term, note that for any round $t$ where $\alphahat_t < \alpha^*$, the buyer's strategy will be equivalent to the equilibrium strategy of the PD game with $\alpha = \tilde{\alpha}$, so the seller's expected utility is $u_S(\tilde{\alpha})$. Thus, $\E\left[\frac 1 T \sum_{t=1}^T U_S(\pi_t) | A_T^C \right] = \frac 1 T \sum_{t=\sqrt{T}}^T \E[U_S(\pi_t) \mid \alpha_t < \alpha^*] +  \frac 1 T \sum_{t=1}^{\sqrt{T}} \E[U_S(\pi_t) \alpha_t < \alpha^*] \geq \frac 1 T \sum_{t=\sqrt{T}}^T u_S(\tilde{\alpha}) +  \frac 1 T \sum_{t=1}^{\sqrt{T}} m = u_S(\tilde{\alpha}) + (m - u_s(\tilde{\alpha}))/ \sqrt{T}$

    Plugging back into the expected average seller's utility yields 
    $\E\left[\frac 1 T \sum_{t=1}^T U_S(\pi_t)\right] \geq \gamma_T m + (1-\gamma_T) [u_s(\tilde{\alpha}) + (m - u_s(\tilde{\alpha}))/ \sqrt{T}]$. Since the seller is playing $\alpha_t \equiv \tilde{\alpha} < \alpha^*$, the consistency of the buyer's beliefs tells us that $\lim_{T \to \infty} \gamma_T = 0$, so $\lim_{T \to \infty} \E\left[\frac 1 T \sum_{t=1}^T U_S(\pi_t)\right] \geq u_S(\tilde{\alpha}) = \text{\stackutility} - \epsilon$.
    Taking $\epsilon$ to 0 gives the desired result. 
\end{proof}

\subsection{Proof of Proposition~\ref{prop:noBetterThanStack}}\label{proof:noBetterThanStack}
\begin{proof} 
Similar to the proof of Proposition~\ref{prop:regretOddsStack}, we will show that due to the consistency property of the beliefs $(\alphahat_t)$, we cannot simultaneously have a high degree of price-discrimination and truthful behavior from buyers. And this will imply that the seller cannot be better than \stackutility asymptotically.  \\

    We will provide the proof for the case $\thetalow \leq \mu \thetahigh$ and the proof for the other case follows similarly. 

    Let us compare the cumulative seller utilities due to a sequence of $(\pi_t)_{t=1}^T$ versus $T \cdot \text{\stackutility}$. Let us denote by ${\pi^p}^*$ the \stackstrat given in Corollary \ref{cor:stackelberg}. Then $\text{\stackutility} = U_S(\truthstrat, {\pi^p}^*)$. In the case of $\thetalow \leq \mu \thetahigh$, this is the strategy where $\pi^p(\thetalow | \slow) = 1$ and $\pi^p(\thetahigh | \shigh) = \alpha^*$. $\sum_{t=1}^T (U_S(\pi_t) - U_S(\pi^*))$ is 

    \newcommand{\size}{m}
    
    \begin{align*}
        \sum_{t \in \mathcal{T}} \left ( U_S(\truthstrat, \pi^p_t) - U_S(\truthstrat, {\pi^p}^*) \right ) + \sum_{t \not \in \mathcal{T}}\left ( U_S(\stratstrat, \pi^p_t) - U_S(\truthstrat, {\pi^p}^*) \right )
    \end{align*}

Note that for all $\pi^p$, $U_S(\stratstrat, \pi^p_t) \leq U_S(\stratstrat, \pdstrat) < U_S(\truthstrat, {\pi^p}^*)$.

    \begin{align*}
        U_S(\truthstrat, {\pi^p}^*) - U_S(\stratstrat, \pdstrat) &= \alpha^* \mu \thetahigh + \alpha^*(1-\mu) \thetalow + (1 - \alpha^*) \thetalow - \thetalow + \mu q^* c_S \\
        &\geq \alpha^* \mu \Delta \theta. \\
        \sum_{t=1}^T (U_S(\pi_t) - U_S(\pi^*)) &\leq \sum_{t \in \mathcal{T}} \left [ U_S(\truthstrat, \pi^p_t) - U_S(\truthstrat, {\pi^p}^*) \right ] - \alpha^* (T-|\mathcal{T}|) \mu \Delta \theta 
        \end{align*}

        \begin{align*}
        \intertext{Note that for any $\pi^p$, the seller's utility when buyers signal truthfully is}
        U_S(\truthstrat, \pi^p) &= \mu \thetahigh \pi^p(\thetahigh | \shigh) + \mu \thetalow \pi^p(\thetalow | \shigh) + (1 - \mu) \thetalow \pi^p(\thetalow | \slow) + 0 \cdot \pi^p(\thetahigh | \slow) \\
        \intertext{Since $\thetalow \leq \mu \thetahigh$, $\pi^{p^*}(\thetalow | \slow) = 1$ and $\pi^{p^*}(\thetahigh | \shigh) = \alpha^*$.}
        U_S(\truthstrat, \pi^p) - U_S(\truthstrat, \pi^{p^*}) &= \mu \thetahigh (\pi^p \thetahigh | \shigh) - \alpha^*) + \mu \thetalow (1 - \pi^p - (1 - \alpha^*)) + (1 - \mu) \thetalow (\pi^p(\thetalow | \slow) - 1) \\
        &= \mu \Delta \theta (\pi^p \thetahigh | \shigh) - \alpha^*) + (1 - \mu) \thetalow (\pi^p(\thetalow | \slow) - 1)
        \end{align*}

        \begin{align*}
        \sum_{t=1}^T (U_S(\pi_t) - U_S(\pi^*)) &\leq \sum_{t \in \mathcal{T}} \left [\mu \Delta \theta (\pi^p(\thetahigh | \shigh) - \alpha^*) - (1-\mu) \thetalow (1 - \pi^p(\thetalow | \slow)) \right ] - \alpha^* (T-|\mathcal{T}|) \mu \Delta \theta
        \intertext{Since $\mu \thetahigh < \thetalow$, $(1-\mu) \thetalow > \mu \Delta \theta$. So, }
        &< \sum_{t \in \mathcal{T}} \mu \Delta \theta (\pi^p_t(\thetahigh | \shigh) - \pi^p_t(\thetalow | \slow) - \alpha^*) - \alpha^* (T-|\mathcal{T}|) \mu \Delta \theta \\
        &= \sum_{t \in \mathcal{T}} \mu \Delta \theta (\pi^p_t(\thetahigh | \shigh) - \pi^p_t(\thetalow | \slow)) - \alpha^* T \mu \Delta \theta \\
        &\leq \mu \Delta \theta \sum_{t \in \mathcal{T}} \alpha_t - \alpha^* T \mu \Delta \theta \quad \text{(By Lemma~\ref{lem:relnPDMeasures})} \\
        \Longrightarrow \lim_{T \rightarrow \infty} \sum_{t=1}^T (U_S(\pi_t) - U_S(\pi^*)) &\leq \mu \Delta \theta \left (  \lim_{T \rightarrow \infty} \frac 1 T \sum_{t \in \mathcal{T}} \alpha_t - \alpha^* \right ) \\
        &\leq 0 \quad \text{(By Lemma~\ref{lem:consistencySimultaneousPDTruthful})}
    \end{align*}

\end{proof}

\end{document}